\RequirePackage{luatex85}

\documentclass[11pt,a4paper]{article}
\usepackage{iftex}

\ifPDFTeX
\usepackage[utf8]{inputenc}
\usepackage[T1]{fontenc}
\fi

\usepackage{amsmath}
\usepackage{amssymb}
\usepackage{amsthm}
\usepackage{thmtools}

\ifPDFTeX
\usepackage{libertine}
\usepackage[libertine]{newtxmath} % must load after ams packages
\usepackage{MnSymbol}
\usepackage[scaled=0.96]{zi4} % use Inconsolata as the typewriter font
\else
\usepackage{fontspec}
\usepackage{unicode-math}
\fi

\usepackage[DIV=11]{typearea} % large value = less margin, recommended values: 10pt: 8, 11pt: 10, 12pt: 12

\usepackage{microtype}

\usepackage{mathtools}

\usepackage[procnumbered,ruled,vlined,linesnumbered,algo2e]{algorithm2e}
\usepackage{tikz}
\usepackage{enumerate}
\usepackage{caption}
\usepackage{subcaption}

\usepackage[
	style=alphabetic,
	backref=true,
	doi=false,
	url=false,
	maxcitenames=3,
	mincitenames=3,
	maxbibnames=10,
	minbibnames=10,
	backend=bibtex8,
	sortlocale=en_US
]{biblatex}

\usepackage[ocgcolorlinks]{hyperref} % Option ocgcolorlinks makes links only colored on screen and not on printouts
\usepackage{cleveref}

\makeatletter
\g@addto@macro\bfseries{\boldmath}
\makeatother

\makeatletter
\g@addto@macro\mdseries{\unboldmath}
\g@addto@macro\normalfont{\unboldmath}
\g@addto@macro\rmfamily{\unboldmath}
\g@addto@macro\upshape{\unboldmath}
\makeatother

\DeclareCiteCommand{\citem}
    {}
    {\mkbibbrackets{\bibhyperref{\usebibmacro{postnote}}}}
    {\multicitedelim}
    {}

\renewcommand*{\multicitedelim}{\addcomma\space}

\AtEveryCitekey{\clearfield{note}}

\newcommand{\myhref}[1]{%
  \iffieldundef{doi}
    {\iffieldundef{url}
       {#1}
       {\href{\strfield{url}}{#1}}}
    {\href{http://dx.doi.org/\strfield{doi}}{#1}}%
}

\DeclareFieldFormat{title}{\myhref{\mkbibemph{#1}}}
\DeclareFieldFormat
  [article,inbook,incollection,inproceedings,patent,thesis,unpublished]
  {title}{\myhref{\mkbibquote{#1\isdot}}}

\makeatletter
\AtBeginDocument{%
    \newlength{\temp@x}%
    \newlength{\temp@y}%
    \newlength{\temp@w}%
    \newlength{\temp@h}%
    \def\my@coords#1#2#3#4{%
      \setlength{\temp@x}{#1}%
      \setlength{\temp@y}{#2}%
      \setlength{\temp@w}{#3}%
      \setlength{\temp@h}{#4}%
      \adjustlengths{}%
      \my@pdfliteral{\strip@pt\temp@x\space\strip@pt\temp@y\space\strip@pt\temp@w\space\strip@pt\temp@h\space re}}%
    \ifpdf
      \typeout{In PDF mode}%
      \def\my@pdfliteral#1{\pdfliteral page{#1}}% I don't know why % this command...
      \def\adjustlengths{}%
    \fi
    \ifxetex
      \def\my@pdfliteral #1{}% isn't equivalent to this one
      \def\adjustlengths{\setlength{\temp@h}{-\temp@h}\addtolength{\temp@y}{1in}\addtolength{\temp@x}{-1in}}%
    \fi%
    \def\Hy@colorlink#1{%
      \begingroup
        \ifHy@ocgcolorlinks
          \def\Hy@ocgcolor{#1}%
          \my@pdfliteral{q}%
          \my@pdfliteral{7 Tr}% Set text mode to clipping-only
        \else
          \HyColor@UseColor#1%
        \fi
    }%
    \def\Hy@endcolorlink{%
      \ifHy@ocgcolorlinks%
        \my@pdfliteral{/OC/OCPrint BDC}%
        \my@coords{0pt}{0pt}{\pdfpagewidth}{\pdfpageheight}%
        \my@pdfliteral{F}% Fill clipping path (the url's text) with
        \my@pdfliteral{EMC/OC/OCView BDC}%
        \begingroup%
          \expandafter\HyColor@UseColor\Hy@ocgcolor%
          \my@coords{0pt}{0pt}{\pdfpagewidth}{\pdfpageheight}%
          \my@pdfliteral{F}% Fill clipping path (the url's text)
        \endgroup%
        \my@pdfliteral{EMC}%
        \my@pdfliteral{0 Tr}% Reset text to normal mode
        \my@pdfliteral{Q}%
      \fi
      \endgroup
    }%
}
\makeatother

\colorlet{DarkRed}{red!50!black}
\colorlet{DarkGreen}{green!50!black}
\colorlet{DarkBlue}{blue!50!black}

\hypersetup{
	linkcolor = DarkRed,
	citecolor = DarkGreen,
	urlcolor = DarkBlue,
	bookmarksnumbered = true,
	linktocpage = true
}

\declaretheorem[numberwithin=section]{theorem}
\declaretheorem[numberlike=theorem]{lemma}

\declaretheorem[numberlike=theorem]{definition}

\crefname{algorithm}{Procedure}{Procedures}
\Crefname{algorithm}{Procedure}{Procedures}

\allowdisplaybreaks[1]

\DontPrintSemicolon
\SetKw{KwAnd}{and}
\SetProcNameSty{textsc}
\SetFuncSty{textsc}

\newcommand{\cC}{\mathcal{C}\xspace}
\newcommand{\cA}{\mathcal{A}\xspace}
\newcommand{\sssp}{{\sf SSSP}\xspace}
\newcommand{\stsp}{\textsf{stSP}\xspace}
\newcommand{\dist}{d}
\newcommand{\hdist}{\mathit{hop}} 
\newcommand{\distest}{\hat{d}}
\newcommand{\diam}{{D}\xspace}
\newcommand{\wdiam}{\mathit{WD}}
\newcommand{\ball}{B}
\newcommand{\CC}{\mathcal{L}}
\newcommand{\bunch}{\mathit{B}}
\newcommand{\clust}{\mathit{C}}
\newcommand{\eps}{\epsilon}
\newcommand{\congest}{\textsf{CONGEST} }
\newcommand{\poly}{\operatorname{poly}}
\newcommand{\polylog}{\operatorname{polylog}}

\title{A Deterministic Almost-Tight Distributed Algorithm for Approximating Single-Source Shortest Paths\thanks{Accepted to \emph{SIAM Journal on Computing}. A preliminary version of this paper was presented at the \emph{48th ACM Symposium on Theory of Computing (STOC 2016)}. The research leading to these results has received funding from the European Research Council under the European Union's Seventh Framework Programme (FP7/2007-2013) / ERC grant agreement No.\ 340506 and ERC grant agreement No.\ 317532. Supported by Swedish Research Council grant 2015-04659 ``Algorithms and Complexity for Dynamic Graph Problems.'' M.\ Henzinger's work was done in part while visiting the Simons Institute for the Theory of Computing. S.\ Krinninger's work was done in part while at the University of Vienna, Austria, while visiting the Simons Institute for the Theory of Computing, and while at the Max Planck Institute for Informatics, Saarland Informatics Campus, Germany.}}

\author{Monika Henzinger\thanks{University of Vienna, Faculty of Computer Science, Austria} \and Sebastian Krinninger\thanks{Department of Computer Sciences, University of Salzburg, Austria.} \and Danupon Nanongkai\thanks{Department of Theoretical Computer Science, KTH Royal Institute of Technology, Sweden.}}

\date{}

\hypersetup{
	pdftitle = {A Deterministic Almost-Tight Distributed Algorithm for Approximating Single-Source Shortest Paths},
	pdfauthor = {Monika Henzinger, Sebastian Krinninger, Danupon Nanongkai}
}

\begin{document}
\maketitle
\begin{abstract}
We present a deterministic $(1+o(1))$-approximation $(n^{1/2+o(1)}+\diam^{1+o(1)})$-time algorithm for solving the single-source shortest paths problem on distributed weighted networks (the \congest model); here $n$ is the number of nodes in the network, $\diam$ is its (hop) diameter, and edge weights are positive integers from $ 1 $ to $ \poly (n) $.
This is the first nontrivial deterministic algorithm for this problem.
It also improves (i) the running time of the randomized $(1+o(1))$-approximation $\tilde O(\sqrt{n} \diam^{1/4} + \diam)$-time\footnote{Throughout, we use $\tilde O (\cdot)$ to hide polylogarithmic factors in $n$.} algorithm of Nanongkai~\citem[STOC 2014]{Nanongkai-STOC14} by a factor of as large as~$n^{1/8}$, and (ii) the $O(\epsilon^{-1}\log \epsilon^{-1})$-approximation factor of Lenzen and Patt-Shamir's $\tilde O(n^{1/2+\epsilon}+\diam)$-time algorithm~\citem[STOC 2013]{LenzenP_stoc13} within the same running time.
Our running time matches the known time lower bound of $\Omega(\sqrt{n / \log{n}} + \diam)$~\citem[Elkin, STOC 2004]{Elkin06} up to subpolynomial factors, thus essentially settling the status of this problem which was raised at least a decade ago \citem[Elkin, SIGACT News 2004]{Elkin04}.
It also implies a $(2+o(1))$-approximation $(n^{1/2+o(1)}+\diam^{1+o(1)})$-time algorithm for approximating a network's weighted diameter which almost matches the lower bound by Holzer and Pinsker~\citem[OPODIS 2015]{HolzerP15}.

In achieving this result, we develop two techniques which might be of independent interest and useful in other settings: (i) a deterministic process that replaces the ``hitting set argument'' commonly used for shortest paths computation in various settings, and (ii) a simple, deterministic construction of an {\em $(n^{o(1)}, o(1))$-hop set} of size $ n^{1+o(1)} $.
We combine these techniques with many distributed algorithmic techniques, some of which are from problems that are not directly related to shortest paths, e.g., ruling sets \citem[Goldberg et al., STOC 1987]{GoldbergPS88}, source detection \citem[Lenzen and Peleg, PODC 2013]{LenzenP_podc13}, and partial distance estimation \citem[Lenzen and Patt-Shamir, PODC 2015]{LenzenP14a-distance}.
Our hop set construction also leads to single-source shortest paths algorithms in two other settings: (i) a $(1+o(1))$-approximation $ n^{o(1)} $-time algorithm on {\em congested cliques}, and (ii) a $(1+o(1))$-approximation  $ n^{o(1)} $-pass $ n^{1+o(1)} $-space {\em streaming} algorithm.
The first result answers an open problem in~\citem[Nanongkai, STOC 2014]{Nanongkai-STOC14}.
The second result partially answers an open problem raised by McGregor in 2006~\citem[{\tt sublinear.info}, Problem 14]{sublinear_open_14}.

\end{abstract}

\section{Introduction}\label{sec:intro}

In the area of {\em distributed graph algorithms} we study the complexity required for a network to compute its own topological properties, such as minimum spanning tree, maximum matching, or distances between nodes.
A fundamental question in this area that has been studied for many years is how much {\em time complexity} is needed to solve a problem in the so-called {\em \congest model} (e.g., \cite{GarayKP98,PelegR00,Elkin06,DasSarmaHKKNPPW12,LenzenP_stoc13}).
In this model (see \Cref{sec:prelim} for details), a network is modeled by a weighted undirected graph $G$, where each node represents a processor that initially only knows its adjacent edges and their weight, and nodes must communicate with each other over \emph{bounded-bandwidth} links to discover global topological properties of the network. The communication between nodes is carried out in {\em rounds}, where in each round each node can send a small, logarithmic-sized message to each neighbor. The time complexity is measured as the number of rounds needed to finish the task. It is usually measured by~$n$, the number of nodes in the network, and $\diam$, the diameter of the communication network (when edge weights are omitted). Typically, $\diam\ll n$. 

In this paper, 
we consider the problem of {\em approximating single-source shortest paths} (\sssp). 
In this problem, a node $s$ is marked as the {\em source node}, and the goal is for every node to know how far it is from $s$. 
The unweighted version---the {\em breadth-first search tree} computation---is one of the most basic tools in distributed computing, and is well known to require $\Theta(\diam)$ time (see, e.g., \cite{Peleg00_book}). In contrast, the only available solution for the weighted case is the distributed version of the Bellman--Ford algorithm \cite{Bellman58,Ford56}, which takes $O(n)$ time to compute an exact solution. 
In 2004, Elkin~\cite{Elkin04} raised the question of whether distributed {\em approximation} algorithms can help in improving this time complexity and showed that any (randomized) $\alpha$-approximation algorithm requires $\Omega(\sqrt{n/(\alpha \log n))} + \diam)$ time~\cite{Elkin06}, which in particular means $\Omega(\sqrt{n  / \log n} + \diam)$ time for any constant-factor approximation.
Das~Sarma~et~al.~\cite{DasSarmaHKKNPPW12} (building on \cite{PelegR00,KorKP13}) later showed that even any (randomized) $\poly(n)$-approximation algorithm requires $\Omega(\sqrt{n} / \log n + \diam)$ time. This lower bound  was later shown to hold even for quantum algorithms \cite{ElkinKNP14}. 

Since running times of the form $\tilde O(\sqrt{n} + \diam)$ show up in many distributed algorithms (e.g., minimum spanning tree~\cite{KuttenP98,PelegR00}, connectivity~\cite{Thurimella97,PritchardT11}, and minimum cut~\cite{NanongkaiS14_disc,GhaffariK13}), it is natural to ask whether the lower bound of~\cite{Elkin06} can be matched. The first answer to this question is a randomized $O(\epsilon^{-1}\log \epsilon^{-1})$-approximation $\tilde O(n^{1/2+\epsilon}+\diam)$-time algorithm by Lenzen and Patt-Shamir~\cite{LenzenP_stoc13}\footnote{Note that the result of Lenzen and Patt-Shamir in fact solves a more general problem.}. The running time of this algorithm is nearly tight if we are satisfied with a large approximation ratio. For a small approximation ratio, Nanongkai~\cite{Nanongkai-STOC14} presented a randomized $(1+o(1))$-approximation $\tilde O(\sqrt{n} \diam^{1/4} + \diam)$-time algorithm. 
The running time of this algorithm is nearly tight when $\diam$ is small, but can be close to $\tilde \Theta(n^{2/3})$ even when $\diam = o(n^{2/3})$. 
This created a rather unsatisfying situation: First, one has to sacrifice a large approximation factor in order to achieve the near-optimal running time, and to achieve a $(1+o(1))$-approximation factor, one must pay an additional running time of $\diam^{1/4}$ which could be as far from the lower bound as $n^{1/8}$ when $\diam$ is large. Because of this, the question of whether we can close the gap between upper and lower bounds for the running time of $(1+o(1))$-approximation algorithms was left as the main open problem in \cite[Problem 7.1]{Nanongkai-STOC14}. 
Second, and more importantly, both these algorithms are randomized. Given that designing deterministic algorithms is an important issue in distributed computing, this leaves the important open problem of whether there is a deterministic algorithm that is faster than the Bellman--Ford algorithm, i.e., that runs in {\em sublinear time}.

\subsection{Our Results} In this paper, we resolve the two issues above. We present a deterministic $(1+o(1))$-approximation $(n^{1/2+o(1)}+\diam^{1+o(1)})$-time algorithm for this problem (the $o(1)$ term in the approximation ratio hides a $1/\polylog n$ factor, and the $o(1)$ term in the running time hides an $O(\sqrt{\log{\log{n}} / \log{n}})$ factor). Our algorithm almost settles the status of this problem as its running time matches the lower bound of Elkin~\cite{Elkin06} up to an $ n^{o(1)} $ factor.

Since an $\alpha$-approximate solution to \sssp gives a $2\alpha$-approximate value of the network's {\em weighted diameter} (cf. \Cref{sec:prelim}), our algorithm can $(2+o(1))$-approximate the weighted diameter within the same running time. Previously, Holzer and Pinsker~\cite{HolzerP15} (building on~\cite{HolzerW12}) showed that for any $\epsilon>0$, a $(2-\epsilon)$-approximation algorithm for this problem requires $\tilde \Omega(n)$ time. Thus, the approximation ratio provided by our algorithm cannot be significantly improved without  increasing the running time.

Using the same techniques, we also obtain a deterministic $(1+o(1))$-approximation $ n^{o(1)} $-time algorithm for the special case of {\em congested clique}, where the underlying network is fully connected.
This gives a positive answer to Problem~7.5 in \cite{Nanongkai-STOC14}. Previous algorithms solved this problem exactly in time $\tilde O(\sqrt{n})$~\cite{Nanongkai-STOC14} and $\tilde O(n^{1/3})$~\cite{CensorHillelKKLPS15}, respectively, and $(1+o(1))$-approximately in time $O(n^{0.158})$~\cite{CensorHillelKKLPS15}\footnote{With this running time, \cite{CensorHillelKKLPS15} can in fact solve the all-pairs shortest paths problem. See also \cite{LeGall16} for further developments in the direction of \cite{CensorHillelKKLPS15}.}.
We can also compute a $(2+o(1))$-approximation of the weighted diameter within the same running time.
The lower bound of Holzer and Pinsker~\cite{HolzerP15} also applies in this setting: Computing a $(2-o(1))$-approximation of the diameter requires $ \tilde \Omega(n) $ time in the worst case.

Our techniques also lead to a (nondistributed) {\em streaming algorithm} for $(1+o(1))$-approximate \sssp, where the edges are presented in an arbitrary-order stream and an algorithm with limited space (preferably $\tilde O(n)$) reads the stream in {\em passes} to determine the answer (see, e.g., \cite{McGregor14} for a recent survey).  
It was known that $\tilde O(n)$ space and one pass are enough to compute an $O(\log n/ \log \log n)$-spanner and therefore approximate all distances up to a factor of $O(\log n/ \log \log n)$  \cite{FeigenbaumKMSZ08} (see also \cite{FeigenbaumKMSZ05,Baswana08,ElkinZ06,Elkin11}). This almost matches a lower bound which holds even for the
$s$-$t$-shortest path problem (\stsp), where we just want to compute the distance between two specific nodes $s$ and~$t$~\cite{FeigenbaumKMSZ08}.
On unweighted graphs one can compute $ (1+\epsilon, \beta) $-spanners in $ \beta $ passes and $ O (n^{1+1/k}) $ space~\cite{ElkinZ06} (for some integer $ \beta $ depending on $ k $ and $ \epsilon $), and get $ (1 + \epsilon) $-approximate \sssp in a total of $ O (\beta / \epsilon) $ passes.
In 2006, McGregor raised the question of whether we can solve \stsp better with a larger number of passes (see \cite{sublinear_open_14}). Very recently Guruswami and Onak \cite{GuruswamiO13} showed that any $p$-pass algorithm on unweighted graphs requires $\tilde \Omega(n^{1+\Omega(1/p)}/O(p))$ space. This does not rule out, for example, an $O(\log n)$-pass $\tilde O(n)$-space algorithm. Our algorithm, which solves the more general  \sssp problem, gets close to this: It takes $ n^{o(1)} $ passes and $ n^{1+o(1)} $ space.

In all of these models, we have formulated our algorithms to compute $ (1 + o(1)) $-approximate \sssp.
More generally, we can, for any $ 0 < \epsilon \leq 1 $, compute a $(1 + \epsilon)$-approximation taking $ {(\sqrt{n} + \diam)} \cdot 2^{O(\sqrt{\log{n} \log{(\epsilon^{-1} \log{n})}})} $ rounds in the \congest model, $ 2^{O(\sqrt{\log{n} \log{(\epsilon^{-1} \log{n})}})} $ rounds in the congested clique model, and $ 2^{O(\sqrt{\log{n} \log{(\epsilon^{-1} \log{n})}})} $ passes with $ n \cdot 2^{O(\sqrt{\log{n} \log{(\epsilon^{-1} \log{n})}})} $ space in the streaming model, respectively.
We provide the necessary details for deriving these numbers in \Cref{sec:hop_set}, but omit them later on for the sake of succinctness.
Our algorithm requires each node to internally store and approximately solve hitting set instances, which can be done in linear time by a greedy algorithm~\cite{Johnson74,AusielloDP80}.
In the \congest model these instances have size $ \sqrt{n} \cdot 2^{O(\sqrt{\log{n} \log{(\epsilon^{-1} \log{n})}})} $, whereas in the congested clique and the multipass streaming model these instances have size $ n \cdot 2^{O(\sqrt{\log{n} \log{(\epsilon^{-1} \log{n})}})} $, respectively.
We assume throughout that the edge weights are positive integers in the range $ \{ 1, \dots, W \} $ where $ W $ is polynomial in $ n $.
More generally, for $ W $ of arbitrary size, all of the above asymptotic bounds need to be multiplied by the factor $ \log{W} $.

\subsection{Overview of Techniques}

Our algorithm builds on two independent contributions: (1) a deterministic process to hit long paths for constructing an overlay network and (2) a deterministic hop set construction for the overlay network.

\subsubsection{Deterministic Path Hitting}

Our crucial new technique is a deterministic process that can replace the following ``path hitting'' argument: For any~$c$, if we pick $\tilde \Theta(c)$ nodes uniformly at random as {\em centers} (typically $c=\sqrt{n}$), then a shortest path containing $n/c$ edges will contain a center with high probability. This allows us to create shortcuts between centers---where we replace each path of length $n/c$ between centers by an edge of the same length---and focus on computing shortest paths between centers. 
This argument has been repetitively used to solve shortest paths problems in various settings (e.g., \cite{UllmanY91,HenzingerK95,DemetrescuI06,BaswanaHS07,RodittyZ11,Sankowski05,DemetrescuFI05,DemetrescuFR09,Madry10,Bernstein13,LenzenP_stoc13,Nanongkai-STOC14}).
In the sequential model a set of centers of size $\tilde \Theta(c)$ can be found deterministically with the greedy hitting set heuristic once the shortest paths containing $n/c$ edges are known~\cite{Zwick02,King99}.
We are not aware of any nontrivial deterministic process that can achieve the same effect in the distributed setting. The main challenge is that the greedy process is heavily sequential, as the selection of the next node depends on all previous nodes, and is thus hard to implement efficiently in the distributed setting\footnote{We note that the algorithm of King \cite{King99} for constructing a blocker can be viewed as an efficient way to greedily pick a hitting set by efficiently computing the scores of nodes. The process is as highly sequential as other greedy heuristics.}.

\paragraph{Approximate Path Hitting via Node Types}

In this paper, we develop a new deterministic process to pick $\tilde \Theta(c)$ centers. The key new idea is to carefully divide nodes into $ O(\log{n})$ {\em types}. Roughly speaking, we associate each type $t$ with a value $w_t$ and make sure that the following properties hold: (i) every path $\pi$ with $\Omega(n/c)$ edges and weight $\Theta(w_t)$ contains a node of type~$t$, and (ii) there is a set of $O(n/c)$ {\em centers of type $t$} such that every node of type $t$ has at least one center at distance $o(w_t)$. We define the set of centers to be the collection of centers of all types. 
The two properties together guarantee that every long path will be {\em almost} hit by a center: For every path~$\pi$ containing at least $n/c$ edges, there is a center whose distance to some node in $\pi$ is $o(w(\pi))$, where $w(\pi)$ is the total weight of~$ \pi $. 
This is already sufficient for us to focus on computing shortest paths only between centers as we would have done after picking centers using the path hitting argument. 
To the best of our knowledge, such a deterministically constructed set of centers that almost hits {\em every} long path was not known to exist before. The process itself is not constrained to the distributed setting and thus might be useful for derandomizing other algorithms that use the path hitting argument.

\paragraph{Distributed Implementation}

To implement the above process in the distributed setting, we use the {\em source detection} algorithm of Lenzen and Peleg~\cite{LenzenP_podc13} to compute the type of each node. We then use the classic {\em ruling set} algorithm of Goldberg, Plotkin, and Shannon~\cite{GoldbergPS88} to compute the set of centers of each type that satisfies the second property above. (A technical note: We also need to compute a bounded-depth shortest-path tree from every center. In \cite{Nanongkai-STOC14}, this was done using the {\em random delay} technique. We also derandomize this step by adapting the {\em partial distance estimation algorithm} of Lenzen and Patt-Shamir \cite{LenzenP14a-distance}.)

\subsubsection{Deterministic Hop Set Construction}

Another tool, which is the key to the improved running time, is a new \emph{hop set} construction. An $(h, \epsilon)$-hop set of a graph $G=(V, E)$ is a set $F$ of weighted edges such that the distance between any pair of nodes in $G$ can be $(1+\epsilon)$-approximated by their $h$-hop distance (given by a path containing at most $h$ edges) on the graph $H=(V, E\cup F)$ (see \Cref{sec:prelim} for details).
The notion of hop set was defined by Cohen \cite{Cohen00} in the context of parallel computing, although it had been used implicitly earlier, e.g., \cite{UllmanY91,KleinS97} (see \cite{Cohen00} for a detailed discussion). 
The previous \sssp algorithm \cite{Nanongkai-STOC14} was able to construct an $(n/k, 0)$-hop set of size $kn$, for any integer $k\geq 1$, as a subroutine (in~\cite{Nanongkai-STOC14} this was called {\em shortest paths diameter reduction}\footnote{This follows the notion of shortest paths diameter used earlier in distributed computing \cite{KhanP08}}). In this paper, we show that this subroutine can be replaced by the construction of an $(n^{o(1)}, o(1))$-hop set of size $ n^{1+o(1)} $.

Our hop set construction is based on computing {\em clusters}, which is the basic subroutine of Thorup and Zwick's distance oracles~\cite{ThorupZ05} and spanners~\cite{ThorupZ05,ThorupZ06}. It builds on a line of work in dynamic graph algorithms. In \cite{Bernstein09}, Bernstein showed that clusters can be used to construct an $(n^{o(1)}, o(1))$-hop set of size $ n^{1+o(1)} $. Later in \cite{HenzingerKNFOCS14}, we showed that the same kind of hop set can be constructed by using a structure similar to clusters while restricting the shortest-path trees involved to some small distance and that such a construction can be used in the dynamic (more precisely, decremental) setting. The construction, however, has to deal with several complications of the dynamic setting and relies heavily on randomization. In this paper, we build on the same idea, i.e., we construct a hop set using bounded-distance clusters. However, our construction is significantly simplified, to the point that we can treat the cluster computation as a black box. This makes it easy to apply on distributed networks and to derandomize. To this end, we derandomize the construction simply by invoking the deterministic clusters construction of Roditty, Thorup, and Zwick \cite{RodittyTZ05} and observe that it can be implemented efficiently on distributed networks\footnote{We note that the Thorup--Zwick distance oracles and spanners were considered before in the distributed setting (e.g., \cite{LenzenP14a-distance,DasSarmaDP12}).}. 
A similar type of derandomization by locally computing approximate hitting sets has been done before by Holzer and Pinsker~\cite{HolzerP15} when derandomizing Nanongkai's exact hop set construction~\cite{Nanongkai-STOC14} on the congested clique.
We note that it might be possible to use Cohen's hop set construction~\cite{Cohen00} instead of Bernstein's~\cite{Bernstein09} in our application. However, Cohen's construction relies heavily on randomness, and derandomizing it seems significantly more difficult.

\subsection{Recent Developments} After the preliminary version of this paper appeared~\cite{HenzingerKN-STOC16}, Becker et al.~\cite{BeckerKKL16} showed that the $n^{o(1)}$ term in our bounds can be eliminated. 
Elkin and Neiman showed the first construction of sparse hop sets with a constant number of hops \cite{ElkinN-FOCS16}, removing also the inherent dependence on $ \log{W} $, the logarithm of the largest edge weight, in their construction.
The latter carries over to the bounds for approximating \sssp in the congested clique model and the multipass streaming model.
They further showed an application of their hop sets in computing approximate shortest paths from $s$ sources. In particular, using our hop set and a modification of the framework in \cite{Nanongkai-STOC14} and this paper\footnote{More precisely, following Elkin and Neiman~\cite{ElkinN-FOCS16}, one constructs an overlay network of size $\sqrt{sn}$ instead of $\sqrt{n}$ as done in this paper.}, this problem can be solved in $ (sn)^{1/2+o(1)}+D^{o{(1)}} $ rounds. Elkin and Neiman showed a hop set which can be used to reduce the bound to $\tilde O(\sqrt{sn} + D)$ when $s=n^{\Omega(1)}$~\cite{ElkinN-FOCS16}.
In \cite{ElkinN-PODC16}, they also showed further applications of hop sets in the distributed construction of routing schemes. 
It was pointed out by Patt-Shamir (see \cite{Tseng15}) that using our algorithm as a black box, one can simplify and obtain improved running time in the construction of compact routing tables in \cite{LenzenP14a-distance}. (On the other hand, we note that our construction is based on many ideas from \cite{LenzenP14a-distance}.) Our hop set construction also found applications in metric-tree embeddings \cite{FriedrichsL16}. 

\subsection{Organization} We start by introducing notation and the main definition in \Cref{sec:prelim}.
Then in \Cref{sec:hop_set} we explain the deterministic hop set construction, which is based on a variation of Thorup and Zwick's clusters~\cite{ThorupZ05}.
In \Cref{sec:dist algo}, we give our main result, namely the $(1+o(1))$-approximation $ (n^{1/2+o(1)}+D^{1+o(1)})$-time algorithm. In that section we explain the deterministic process for selecting centers mentioned above, as well as how to implement the hop set construction in the distributed setting. 
Finally, our remaining results are proved in \Cref{sec:other results}.

\section{Preliminaries}\label{sec:notation}\label{sec:prelim}

\subsection{Notation}
In this paper, we consider weighted undirected graphs with positive integer edge weights in the range $ \{1, 2, \ldots, W \} $.
We usually assume in the following that $ W = \poly (n) $, i.e., the edge weights are \emph{polynomially bounded}.
For a graph $ G = (V, E) $, $ V $ is the set of nodes and $ E $ is the set of edges.
We denote by $ n := |V| $ and $ m := |E| $ the number of nodes and edges of $ G $, respectively.
For a set of edges $ E $, the weight of each edge $ (u, v) \in E $ is given by a function $ w (u, v, E) $.
If $ (u, v) \notin E $, we set $ w (u, v, E) = \infty $.
We define $ w (u, v, G) := w (u, v, E) $.
Whenever we define a set of edges $ E $ as the union of two sets of edges $ E_1 \cup E_2 $, we set the weight of every edge $ (u, v) \in E $ to $ w (u, v, E) := \min (w (u, v, E_1), w (u, v, E_2)) $.
We denote the weight of a path $ \pi $ in a graph $ G $ by $ w (\pi, G) := \sum_{(u, v) \in \pi} w (u, v, G) $ and the number of edges of $\pi$ by $|\pi|$. 

Given a graph $ G = (V, E) $ and a set of edges $ F \subseteq V^2 $, we define $ G \cup F $ as the graph that has $ V $ as its set of nodes and $ E \cup F $ as its set of edges.
The weight of each edge $ (u, v) $ is given by $ w (u, v, G \cup F) = w (u, v, E \cup F) = \min (w (u, v, E), w (u, v, F)) $.

We denote the distance between two nodes $ u $ and $ v $ in $ G $, i.e., the weight of the shortest path between $ u $ and $ v $, by $ \dist (u, v, G) $.
We define the distance between a node $ u $ and a set of nodes $ A \subseteq V $ by $ \dist (u, A, G) = \min_{v \in A} \dist (u, v, G) $.
For every pair of nodes $ u $ and $ v $ we define distance up to range $ R $ by
\begin{equation*}
\dist (u, v, R, G) =
\begin{cases}
\dist (u, v, G) & \text{if $ \dist (u, v, G) \leq R $} \\
\infty & \text{otherwise} \, .
\end{cases}
\end{equation*}
and for a node $ v $ and set of nodes $ A \subseteq V $ by $ \dist (u, A, R, G) = \min_{v \in A} \dist (u, v, R, G) $. 

For any positive integer~$h$ and any nodes $u$ and~$v$, we define the {\em $h$-hop distance} between $u$ and~$v$, denoted by $\dist^h(u, v, G)$, as the weight of the shortest among all $u$-$v$ paths containing at most $h$ edges. More precisely, let $\Pi^h(u, v)$ be the set of all paths between $u$ and $v$ such that each path $\pi\in \Pi^h(u, v)$ contains at most $h$ edges. Then, $\dist^h (u, v, G) = \min_{\pi\in \Pi^h(u, v)} w(\pi, G)$ if $\Pi^h(u, v)\neq \emptyset$, and $\dist^h (u, v, G) =\infty$ otherwise.

We denote the hop distance between two nodes $ u $ and $ v $, i.e., the distance between $u$ and $v$ when we treat $G$ as an unweighted graph, by $\hdist(u, v, G)$. The {\em hop diameter} of graph $G$ is defined as $\diam(G)=\max_{u, v \in V} \hdist(u, v, G)$. When $G$ is clear from the context, we use $\diam$ instead of $\diam(G)$. We note that this is different from the {\em weighted diameter}, which is defined as $\wdiam(G)=\max_{u, v\in V} \dist(u, v, G)$. Throughout this paper we use ``diameter'' to refer to the hop diameter (as is typically done in the literature; see, e.g., \cite{GarayKP98, KuttenP98, KhanKMPT12, LotkerPR09, GhaffariK13}).
We do not consider superlogarithmic values for the bandwidth $ B $ in this paper.

The following definition formalizes the concept of hop sets introduced by Cohen~\cite{Cohen00}.
\begin{definition}\label{def:hopset}
Given any graph $G=(V, E)$, any integer $h$, and $\epsilon\geq 0$, we say that a set of weighted edges $F$ is an {\em $(h, \epsilon)$-hop set} of $G$ if 
\begin{equation*}
\dist(u, v, G)\leq \dist^h(u, v, H)\leq (1+\epsilon)\dist(u, v, G)
%\label{eq:intro hop set}
\end{equation*} 
for every pair of nodes $ u, v \in V $, where $H=(V, E\cup F)$.
\end{definition}
In this paper we are only interested in $(n^{o(1)}, o(1))$-hop sets of size $ |F| = n^{1+o(1)} $. We refer to them simply as ``hop sets'' (without specifying parameters).

\subsection{\congest Model and Problem Formulation}\label{sec:model}\label{sec:problem}

In the \congest model, a network of processors is modeled by an undirected weighted
graph~$G$, where nodes model the processors and  edges model the bounded-bandwidth links between the processors. 
Nodes are assumed to have unique IDs in the range $\{1, 2, \ldots, \poly(n)\}$ and infinite computational power\footnote{In the algorithms developed in this paper this strong assumption is not necessary as the number of internal computational steps at each node is proportional to the number of messages received in all rounds.} as the primary focus of this model is communication complexity.
We denote by $ \lambda $ the number of bits used to represent each ID, i.e., $ \lambda = O (\log{n}) $.
Each node has limited topological knowledge; in particular, every node~$u$ knows only the IDs of each neighbor $v$ and $w(u, v, G)$, the weight of their connecting edge.
As in \cite{Nanongkai-STOC14}, we assume that edge weights are polynomially bounded, i.e., the largest edge weight of the graph is polynomial in the number of nodes.
This is a typical assumption as it allows us to encode the weight of an edge in one (or a constant number of) messages.

The distributed communication is performed in {\em rounds}. At the beginning of each round, all nodes wake up simultaneously, and then each node $u$ sends an arbitrary message of $B=O(\log n)$ bits through each edge $(u, v)$, and the message will arrive at node $v$ at the end of the round.
For the algorithms presented in this paper, we consider the weaker \emph{broadcast \congest model}, where in every individual round the message sent by each node is the same for all neighbors.
The {\em running time} of a distributed algorithm is the worst-case number of rounds needed to finish a task. It is typically analyzed based on $n$ (the number of nodes) and $\diam$ (the network diameter)~\cite{Peleg00_book}.

\begin{definition}[Single-Source Shortest Paths (\sssp)]\label{def:sssp}
In the {\em single-source shortest paths problem (\/\sssp)}, we are given a weighted network $G$ and a {\em source} node $s$; i.e., each node knows (i) the IDs of its neighbors, (ii) the weight of its incident edges, and (iii) whether it is the source $ s $ or not. We want to find the distance between $s$ and every node $v$ in $G$, denoted by $\dist(s, v, G)$, i.e., we want every node $v$ to know the value of $\dist(s, v, G)$.
In the {\em $ \alpha $-approximate \sssp problem} each node additionally knows the value $ \alpha \geq 1 $, and the goal is for every node $ v $ to know a distance estimate $ \hat{\dist} (s, v) $ such that $ \dist(s, v, G) \leq \hat{\dist} (s, v) \leq \alpha \cdot \dist(s, v, G) $.
\end{definition}

\paragraph{Recovering shortest paths} We note that although we define the problem to be computing the distances, we can easily recover the shortest paths in the sense that every node $u$ knows its neighbor~$v$ that is in the shortest path between $u$ and $s$. 
This is because our algorithm computes a distance estimate that satisfies the following property:
\begin{align}
\mbox{Every node $u\neq s$ has a neighbor $v$ such that } \hat{\dist} (s, v)+w(u, v, G) \leq  \hat{\dist} (s, u),\label{eq:recover shortest paths} 
\end{align}
where $\hat{\dist} (s, v)$ is the approximate distance between $s$ and $v$. For any distance approximation $\hat{\dist} (s, \cdot) $ that satisfies~\eqref{eq:recover shortest paths}, we can recover the approximate shortest paths by assigning $v$ as the intermediate neighbor of $u$ in the approximate shortest path between $u$ and $s$. 

It can be easily checked throughout that the distance estimate that we compute satisfies~\eqref{eq:recover shortest paths}. This is simply because our algorithm always rounds an edge weight $w(u, v, G)$ {\em up} to some value $w'(u, v)$, and computes the approximate distances based on this rounded edge weight. 
For this reason, we can focus only on computing approximate distances in this paper.

\subsection{Toolkit}
In the following we review, in more detail, known results used for designing our algorithm.
The first is a weight-rounding technique~\cite{KleinS97,Cohen98,Zwick02,Bernstein09,Madry10,Bernstein13,Nanongkai-STOC14} for scaling down edge weights at the cost of approximation.
Intuitively, we will use this technique to efficiently compute approximate shortest paths up to a fixed number of hops.
As we will use this technique repeatedly, we give a proof in Appendix~\ref{sec:proof of property of weight rounding} for completeness.

\begin{lemma}[\cite{Nanongkai-STOC14}]\label{thm:property of weight rounding}
	Let $ h \geq 1 $ and let $ G $ be a graph with positive integer edge weights in the range $ \{ 1, \ldots, W \} $.
	For every integer $ 0 \leq i \leq \lfloor \log{(n W)} \rfloor $, set $\rho_i = \frac{\epsilon 2^i}{h}$ and let $G_i$ be the graph with the same nodes and edges as $ G $ and weight $w(u, v, G_i)=\lceil\frac{w(u, v, G)}{\rho_i}\rceil$ for every edge $(u, v)$.
	Then for all pairs of nodes $ u $ and $ v $ and every $ 0 \leq i \leq \lfloor \log{(n W)} \rfloor $
	\begin{equation}
	\rho_i \cdot \dist(u, v, G_i)\geq \dist (u, v, G) \, .\label{eq:apsp approx main one}
	\end{equation}
	Moreover, if $ 2^i \leq \dist^h (u, v, G) \leq 2^{i+1} $, then
	\begin{align}
	\dist(u, v, G_i)&\leq (2 + 2 / \epsilon) h \text{~~~and}\label{eq:apsp approx main two}\\
	\rho_i \cdot \dist(u, v, G_i) &\leq (1+\epsilon) \cdot \dist^h(u, v, G) \, ,\label{eq:apsp approx main three}
	\end{align}
	where $ \dist^h(u, v, G) $ is the $h$-hop distance between $ u $ and $ v $ in $ G $.
\end{lemma}

An important subroutine in our algorithm is a procedure for solving the \emph{source detection} problem \cite{LenzenP_podc13} in which we want to find the $\sigma$ nearest ``sources'' in a set $S$ for every node~$u$, given that $ u $ is at distance at most $ \gamma $ from them.
Ties are broken by ID.
The following definition is from~\cite{LenzenP14a-distance}.

\begin{definition}[$(S, \gamma, \sigma)$-detection]\label{def:source detection}
	Consider a graph $G = (V, E)$, a set of ``sources'' $S \subseteq V$, and parameters $\gamma, \sigma \in \mathbb{N}$.
	For any node $v$, let $ \CC (v, S, \gamma, \sigma, G) $ denote the \emph{proximity list} resulting from ordering the set
	$\{(\dist(u, v, G), u)| u \in S \wedge \dist(u, v, G) \leq \gamma\}$ lexicographically in ascending order, i.e., where
	\begin{multline*}
	(\dist(u, v, G), v) < (\dist(u', v, G), u') \iff \\ (\dist(u, v, G) < \dist(u', v, G)) \vee (\dist(u, v, G) = \dist(u', v, G) \wedge u < u') \, ,
	\end{multline*}
	and restricting the resulting list to the first $ \sigma $ entries.
	The goal of the $(S, \gamma, \sigma)$-detection problem is to compute $ \CC (v, S, \gamma, \sigma, G) $ for every node $v \in V$.
	In the distributed setting we assume that, as part of the input, each node knows $\gamma$, $\sigma$, and whether it is in $S$ or not, and the goal is that every node $v \in V$ knows its list $ \CC (v, S, \gamma, \sigma, G) $.
\end{definition}

Lenzen and Peleg designed a source detection algorithm for unweighted networks in the \congest model~\cite{LenzenP_podc13}.
Their algorithm maintains, for each node, a tentative proximity list, where each entry is a pair consisting of a distance value and a source node in $ S $.
The list of every node $ v $ is initialized with the pair $ (0, v) $, and in every round, each node sends, to all of its neighbors, the smallest entry in its list (according to lexicographic order) that it has not transmitted before.
Upon receiving a pair $ (\delta_s, s) $, a node first checks if there already is some entry $ (\delta_s', s) $ with $ \delta_s' \leq \delta_s + 1 $ in its tentative proximity list, and if not, it adds the pair $ (\delta_s + 1, s) $ to its list (and marks it as not yet transmitted).
Lenzen and Peleg showed that after $ \min{(\gamma, \diam)} + \min{(\sigma, |S|)} $ rounds, the first $ \sigma $ entries in the list maintained by every node $ v $ correspond to $ \CC (v, S, \gamma, \sigma, G) $.
Holzer and Pinsker~\cite{HolzerP15} had two observations about this algorithm.
The first observation is that the guarantees of the algorithm directly carry over to the broadcast \congest model as in every round each node sends the same message to all of its neighbors.
The second observation is that one can also run the algorithm on weighted networks (see also~\cite[proof of Theorem 3.3]{LenzenP14a-distance}) by replacing each edge of some weight~$L$ with an unweighted path of length~$L$ where all the nodes added for some weighted edge $ (u, v) $ are ``simulated'' by either $ u $ or $ v $.
Note that the ``simulated'' nodes are never sources.
Furthermore, the tentative lists of the ``simulated'' nodes do not have to be maintained explicitly.
The following modification of the algorithm for weighted graphs is functionally equivalent to the simulation approach: Every time a node $ u $ wants to send some entry entry $ (\delta_s, s) $ to some neighbor~$ v $ via an edge of weight $ w (u, v, G) $, it delays this message by $ w (u, v) - 1 $ rounds; upon reception, $ v $ first checks if there already is some entry $ (\delta_s', s) $ with $ \delta_s' \leq \delta_s + w (u, v, G) $ in its tentative proximity list, and if not, it adds the pair $ (\delta_s + w (u, v, G), s) $ to its list (and marks it as not yet transmitted).

\begin{theorem}[Implicit in \cite{LenzenP_podc13}]\label{lem:source detection algorithm}
In the broadcast \congest model, there is a deterministic algorithm for solving the $(S, \gamma, \sigma)$-detection problem in $ \min{(\gamma, \wdiam)} + \min{(\sigma, |S|)} $ rounds on weighted networks, where $ \wdiam $ is the weighted diameter.
\end{theorem}

We remark that in an earlier version of this paper we have, for example, in the streaming model, used an additional source detection algorithm of Roditty, Thorup, and Zwick~\cite{RodittyTZ05}.\footnote{Roditty, Thorup, and Zwick~\cite{RodittyTZ05} solve a variant of the source detection problem with $ \gamma = \infty $ in their centralized algorithm for computing distance oracles and spanners deterministically. They essentially reduce the source detection problem to a sequence of \sssp computations on graphs with $ O(n) $ additional nodes and edges. This reduction can be modified in a straightforward way to generalize their algorithm to arbitrary $ \gamma $.}
Using a second algorithm is, however, not essential as the algorithm by Lenzen and Peleg~\cite{LenzenP_podc13} provides all necessary guarantees.

Another subproblem arising in our algorithm is the computation of \emph{ruling sets}.
The following definition was adapted from the recent survey of Barenboim and Elkin \cite[Section 9.2]{BarenboimE13}. 
\begin{definition}[Ruling Set]
	For a (possibly weighted) graph $G = (V, E)$, a subset $U \subseteq V $ of nodes, and 
	a pair of positive integers $\alpha$ and $\beta$, 
	a set $T\subset U$ is an $(\alpha, \beta)$-ruling set for $U$ in $G$ if
	\begin{enumerate}
			\item for every pair of distinct nodes $u, v \in T$, it holds that $\dist(u, v, G) \geq \alpha$, and
			\item for every node $u \in U \setminus T$, there exists a ``ruling'' node $v \in T$, such that $\dist(u, v, G) \leq \beta$.
	\end{enumerate}
\end{definition}
The classic result of Goldberg, Plotkin, and Shannon~\cite{GoldbergPS88} shows that in the distributed setting, for any $c\geq 1$, we can compute a $(c,c \lambda)$-ruling set deterministically in $O(c \log n)$ rounds, where $ \lambda $ is the number of bits used to represent each ID in the network. Since it was not explicitly stated that this algorithm works in the broadcast \congest model, we sketch an implementation of this algorithm in \Cref{sec:ruling set algorithm} (see \cite[Chapter 9.2]{BarenboimE13} and \cite[Chapter 22]{Peleg00_book} for a more detailed algorithm and analysis). 

\begin{theorem}[implicit in \cite{GoldbergPS88}]\label{lem:ruling set algorithm}
In the broadcast \congest model, there is a deterministic algorithm that, for every $ c \geq 1 $, computes a $(c, c \lambda)$-ruling set in $O(c \log n)$ rounds, where $ \lambda $ is the number of bits used to represent each ID in the network.
\end{theorem}

\section{Deterministic Hop Set Construction}\label{sec:hop_set}

In this section we present a deterministic algorithm for constructing an $ (n^{o(1)}, o(1)) $-hop set (see \Cref{def:hopset}).
In \Cref{sec:hop_reduction_with_additive_error} we first give an algorithm with a weaker guarantee that computes a set of edges $ F $ that reduces the number of hops between all pairs of nodes in the following way for some fixed $ \Delta \geq 1 $: If the shortest path has weight $ R $, then using the edges of $ F $, we can find a path with $ \tilde O (R / \Delta) $ edges at the cost of a multiplicative error of $ o(1) $ and an additive error of $ n^{o(1)} \Delta $.
Our algorithm obtains $ F $ by computing the \emph{clusters} of the graph.
We explain clusters and their computation in \Cref{sec:clusters}.
In \Cref{sec:hop_reduction_without_additive_error}, we show how to repeatedly apply the first algorithm for different edge weight modifications to obtain a set of edges $ F $ providing the following stronger hop reduction for all pairs of nodes: If the shortest path has $ h $ hops, then, using the edges of $ F $, we can find a path with $ \tilde O (h / \Delta) $ hops at the cost of a multiplicative error of $ o(1) $ and no additive error.
Finally, in \Cref{sec:computing_hop_set}, we obtain the hop set by repeatedly applying the hop reduction.

\subsection{Deterministic Clusters}\label{sec:clusters}

The basis of our hop set construction is a structure called {\em cluster} introduced by Thorup and Zwick~\cite{ThorupZ05} who used it, e.g., to construct distance oracles~\cite{ThorupZ05} and spanners~\cite{ThorupZ06} of small size.
\begin{definition}\label{def:clusters}
Consider an integer $ p $ such that $ 2 \leq p \leq \log{n} $ and a hierarchy $ \mathcal{A} $ of sets of nodes $ (A_i)_{0 \leq i \leq p} $ such that $ A_0 = V $, $ A_p = \emptyset $, and $ A_0 \supseteq A_1 \supseteq \dots \supseteq A_p $. 
We say that a node $ v $ has \emph{priority $ i $} if $ v \in A_i \setminus A_{i+1} $ (for $ 0 \leq i \leq p-1 $).
For every node $ v \in V $ the \emph{cluster of $ v $} in $ G $ is defined as
\begin{equation*}
\clust (v, \cA, G) = \{ u \in V \mid \dist (u, v, G) < \dist (u, A_{i+1}, G) \} \, ,
\end{equation*}
where $ i $ is the priority of $ v $.
\end{definition}
In the noncentralized models of computation considered in this paper, the straightforward way of computing clusters as defined above is not efficient enough for our purposes.
We can, however, afford to compute the following restricted clusters.
\begin{definition}
Consider $ p $ and $ \mathcal{A} $ defined as in \Cref{def:clusters} and $ R \geq 0 $.
For every node $ v \in V $, the \emph{restricted cluster up to distance $ R $ of $ v $} in $ G $ is defined as
\begin{equation*}
\clust (v, \cA, G) = \{ u \in V \mid \dist (u, v, G) < \dist (u, A_{i+1}, G) \text{ and } \dist (u, v, G) \leq D \} \, ,
\end{equation*}
where $ i $ is the priority of $ v $.
\end{definition}

\subsubsection{Computing Priorities \texorpdfstring{$\cA$}{A}}\label{sec:priorities}

The performance of our algorithms relies on the total size of the clusters, which in turn depends on how we compute nodes' priorities.
If randomization is allowed, we can use the following algorithm due to Thorup and Zwick \cite{ThorupZ05,ThorupZ06}: Set $ A_0 = V $ and $ A_p = \emptyset $, and for $ 1 \leq i \leq p-1 $ obtain $ A_i $ by picking each node from $ A_{i-1} $ with probability $ ((\ln{n})/n)^{1/p} $.
It can be argued that for $ \mathcal{A} = (A_i)_{0 \leq i \leq p} $ the size of all clusters, i.e., $ \sum_{v \in V} | \clust (v, \cA, G) | $, is $ O (p n^{1 + 1/p}) $ in expectation~\cite{ThorupZ05}.
We now explain how to deterministically compute the priorities of nodes (given by a hierarchy of sets of nodes $ \mathcal{A} = (A_i)_{0 \leq i \leq p} $) such that the total size of the resulting clusters is $ \sum_{v \in V} | \clust (v, \cA, G) | = O (p n^{1 + 1/p}) $.

Thorup and Zwick~\cite{ThorupZ05} introduced the notion of {\em bunches} to analyze the sizes of clusters.
For every node $ u \in V $, we define the bunch and, for every $ 0 \leq i \leq p-1 $, the  $i$-bunch, both restricted to distance $ R $, as follows:
\begin{align*}
\bunch_i (u, \cA, R, G) &= \{ v \in A_i \setminus A_{i+1} \mid \dist (u, v, G) < \dist (u, A_{i+1}, G) \text{ and } \dist (u, v, G) \leq R \} \\
\bunch (u, \cA, R, G) &= \bigcup_{0 \leq i \leq p-1} \bunch_i (u, \cA, R, G) \, .
\end{align*}
The crucial insight is that $ v \in \bunch(u, \cA, R, G) $ if and only if $ u \in \clust (v, \cA, R, G)$.
Thus, it suffices to choose a hierarchy of sets $A_i$ such that $ | \bunch_i (u, \cA, R, G) | \leq O (n^{1/p}) $ for every $ u \in V $ and $ 0 \leq i \leq p-1 $.

Our algorithm for deterministically computing this hierarchy of sets of nodes follows the main idea of Roditty, Thorup, and Zwick~\cite{RodittyTZ05}.
Its pseudocode is given in Procedure~\ref{alg:priorities}.
As a subroutine this algorithm solves a weighted source detection problem, i.e., for suitable parameters $ q $, $ A $, and $ R $, it computes for every node $ v $ the proximity list $ \CC (v, A, R, q, G) $ containing the $ q = \tilde O (n^{1/p}) $ nodes of $ A $ that are closest to $ v $---up to distance $ R $; if there are fewer than $ q $ nodes of $ A $ in distance $ R $ to $ v $, then $ \CC (v, A, R, q, G) $ contains all of them.
Our algorithm for constructing the hierarchy of sets $ (A_i)_{0 \le i \leq p} $ is as follows.
We set $ A_0 = V $ and $ A_p = \emptyset $, and to construct the set $ A_{i+1} $ given the set~$ A_i $ for $ 0 \leq i \leq p-2 $, we first find for each node $ v \in V $ the set $ \CC (v, A_i, R, q, G) $ using a source detection algorithm.
Then we view the collection of sets $ \{ \CC (v, A_i, R, q, G) \}_{v \in V} $ as an instance of the hitting set problem over the universe~$ A_i $, where we want to find a set $ A_{i+1} \subseteq A_i $ of minimum size such that each set $ \CC (v, A_i, R, q, G) $ contains at least one node of $ A_{i+1} $, i.e., a hitting set.
We let $ A_{i+1} $ be an approximate hitting set whose size is within a factor of $ 1 + \ln{n} $ of the optimum produced by the deterministic greedy heuristic (always adding the element that is ``hitting'' the largest number of ``un-hit'' sets)~\cite{Johnson74,AusielloDP80}.
Note that the expensive hitting set computation will later be implemented by performing internal computation\footnote{In principle, internal computation is free in the models considered in this paper and we could thus compute a minimum hitting set exactly. However, we decided to present the algorithm in a way that avoids solving NP-complete problems by internal computation.}; see, for example, \Cref{sec:distributed priorities} for the \congest model implementation.
Following~\cite{ThorupZ05}, we explicitly set $ A_p = \emptyset $ to avoid the introduction of special notation for clusters of the largest priority.
In the following we prove the desired bound on the size of the bunches, which essentially requires us to argue that setting $ q = \tilde O (n^{1/p}) $ is sufficient.

\begin{procedure}
	\caption{Priorities($G$, $p$, $R$)}
	\label{alg:priorities}
	
	\KwIn{Weighted graph $ G = (V, E) $ with positive integer edge weights, number of priorities $ p \geq 2 $, distance range $ R \geq 1 $}
	\KwOut{Hierarchy of sets $ \cA = (A_i)_{0 \leq i \leq p} $}
	
	\BlankLine

	$ q \gets \lceil 2 n^{1/p} \ln{(3 n)} (1 + \ln n) \rceil $\;
	$ A_0 \gets V $\;
	\For{$ i = 0 $ \KwTo $ p-2 $}{
		Compute $ \CC (v, A_i, R, q, G) $ for every node $ v \in V $ using a source detection algorithm\;
		$ \cC \gets \emptyset $\;
		\ForEach{$v \in V$}{
			\lIf{$ | \CC (v, A_i, q, R, G) | = q $}{
				$ \cC \gets \cC \cup \{ \CC (v, A_i, R, q, G)\} $
			}
		}
		Compute a $ (1 + \ln n) $-approximate minimum hitting set $ A_{i+1} \subseteq A_i $ using a greedy heuristic\;
	}
	$ A_p \gets \emptyset $\;
	
	\KwRet{$ \cA := (A_i)_{0 \leq i \leq p} $}\;
\end{procedure}

\begin{lemma}[Implicit in~\cite{RodittyTZ05}]\label{lem:existence of hitting set}
Given a finite collection of sets $ \cC = \{ S_1, \ldots, S_k \} $ over a universe $ U $ and a parameter $ x \geq 1 $ such that $ | S_j | \geq 2 x \ln{(3k)} $ for all $ 1 \leq j \leq k $, there exists a hitting set $ T \subseteq U $ of size at most $ | T | \leq |U| / x $ such that $ T \cap S_j \neq \emptyset $ for all $ 1 \leq j \leq k $.
\end{lemma}
We give a proof of \Cref{lem:existence of hitting set} in \Cref{sec:proof of existence of hitting set} for completeness.

\begin{lemma}\label{lem:priorities}
Given a weighted graph $ G = (V, E) $ with positive integer edge weights and parameters $ p \geq 2 $ and $ R \geq 1 $, Procedure~\ref{alg:priorities} computes a hierarchy $ \cA = (A_i)_{0 \leq i \leq p} $ of sets of nodes such that
\begin{equation*}
\sum_{u \in V} | \clust (u, \cA, R, G) | = \sum_{u \in V} | \bunch (u, \cA, R, G) | = O (p n^{1 + 1/p} \log{n}) \, .
\end{equation*}
\end{lemma}

\begin{proof}
We first show by induction that $ |A_i| \leq n^{1-i/p} $ for all $ 0 \leq i \leq p-1 $.
If $ i = 0 $, the claim is trivially true because we set $ A_0 = V $.
We now assume the induction hypothesis $ |A_i| \leq n^{1-i/p} $ and argue that $ |A_i| \leq n^{1-(i+1)/p} $.
Our algorithm computes a $ (1 + \ln{n}) $-approximate hitting set of the collection of sets $ \cC $ containing each set $ \CC (v, A_i, q, R, G) $ of size $ q $ (i.e., not considering all such sets that have size strictly less than $ q $).
Let $ k \leq n $ be the number of sets contained in $ \cC $.
As each set in $ \cC $ has size $ q = \lceil 2 n^{1/p} \ln{(3 n)} (1 + \ln n) \rceil \geq 2 x \ln (3q) $ for $ x = n^{1/p} (1 + \ln n) $, we know by \Cref{lem:existence of hitting set} that there is a hitting set $ A' $ for $ \cC $ of size at most $ |A_{i}| / x = |A_{i}| / (n^{1/p} (1 + \ln{n})) \leq n^{1-i/p} / (1 + \ln{n}) $ and thus the minimum hitting set $ A_{i+1} $ computed by the greedy heuristic has size at most $ (1 + \ln{n}) |A'| \leq n^{1-i/p} $.
Note that each set $ \CC (v, A_i, q, R, G) $ might have been empty, and in this case the algorithm would have computed $ A_{i+1} = \emptyset $, the trivial hitting set.

We now show that for every node $ u \in V $ and every $ 0 \leq i \leq p-1 $, $ \bunch_i (u, \cA, R, G) \leq q = O(n^{1/p} \log{n}) $, which immediately implies the desired bound on the total size of the bunches and clusters.
We argue by a case distinction that $ \bunch_i (u, \cA, R, G) \subseteq \CC(u, A_i, R, q, G) $ and thus, by the definition of the set of the $ q $ closest nodes in~$ A_i $, $ | \bunch_i (v, \cA, R, G) | \leq | \CC(u, A_i, R, q, G) | \leq q $.
If $ | \CC(u, A_i, R, q, G) | < q $, then clearly $ \CC(u, A_i, R, q, G) = \{ v \in A_i \mid \dist (u, v, G) \leq R \} \supseteq \bunch_i (v, \cA, R, G) $.
Otherwise we have $ | \CC(u, A_i, R, q, G) | = q $, and as the algorithm computed a suitable hitting set, we have $ \bunch_i (u, \cA, R, G) \subseteq \CC(u, A_i, R, q, G) $.
\end{proof}

As an alternative to the algorithm proposed above, the hitting sets can also be computed with the deterministic algorithm of Roditty, Thorup, and, Zwick~\cite{RodittyTZ05} which produces so-called ``early hitting sets.''
For this algorithm we have to set $ q = O(n^{1/p} \log{n}) $ and obtain slightly smaller clusters of total size $ O (p n^{1 + 1/p}) $.
However, since the logarithmic factors are negligible for our purpose, we have decided to present the simpler algorithm above.

\subsubsection{Computing Clusters}

Given the priorities of the nodes, the clusters can be computed by finding, for every priority $ i $ and for every node $ v $ of priority $ i $, the shortest paths up to nodes whose distance to~$ v $ is more than (or equal to) their distance to nodes of priority more than $ i $.
In the pseudocode of Procedure~\ref{alg:clusters} we formulate this algorithm as a variant of weighted breadth-first search.
We will not analyze the performance of this algorithm at this point since it depends on the models of computation that simulate it (see \Cref{sec:priorities} and \Cref{sec:other results} for implementations of the algorithm and performance analyses).

\begin{procedure}
	\caption{Clusters($G$, $p$, $R$)}
	\label{alg:clusters}
	
	\KwIn{Weighted graph $ G = (V, E) $ with positive integer edge weights, number of priorities $ p \geq 2 $, distance range $ R \geq 1 $}
	\KwOut{Clusters of $ G $ as specified in \Cref{thm:clusters}}
	
	\BlankLine
	
	$ \cA = (A_i)_{0 \leq i \leq p} \gets $ \Priorities{$G$, $p$, $R$}\;
	For each $ 1 \leq i \leq p-1 $ and every node $ v  \in V $ compute $ \dist (v, A_i, R, G) $\;
	
	\BlankLine
	
	\ForEach(\tcp*[f]{Compute cluster of every node}){$ u \in V $}{
		\tcp{Initialization}
		Let $ i $ be the priority of $ u $, i.e., $ u \in A_i \setminus A_{i+1} $\;
		\lForEach{$ v \in V$}{
			$ \delta (u, v) \gets \infty $
		}
		$ \delta (u, u) \gets 0 $\;
		$ \clust (u) \gets \emptyset $\;
		\tcp{Iteratively add nodes to cluster}
		\For{$ L = 0 $ \KwTo $ R $}{
			\ForEach{node $ v $ with $ \delta (u, v) = L $}{
				\tcp{Check if $ v $ joins cluster of $ u $ at current level}
				\If{$ \delta (u, v) < \dist (v, A_{i+1}, R, G) $}{
					$ \clust (u) \gets \clust (u) \cup \{ v \} $\;
					\ForEach(\tcp*[f]{Update neighbors of $ v $}){$ (v, w) \in E $}{
						$ \delta' (u, w) \gets (w (v, w, G) + \delta (u, v)) $\;
						\lIf{$ \delta' (u, w) < \delta (u, w) $}{
							$ \delta (u, w) \gets \delta' (u, w) $
						}
					}
				}
			}
		}
	}
	\KwRet{$ (\clust  (v), \delta (v, \cdot))_{v \in V} $}\;
\end{procedure}

We summarize our guarantees with the following theorem.

\begin{theorem}\label{thm:clusters}
Given a weighted graph $ G = (V, E) $ with positive integer edge weights and parameters $ p \geq 2 $ and $ R \geq 1 $, Procedure~\ref{alg:clusters} computes a hierarchy of sets $ \cA = (A_i)_{0 \leq i \leq p} $, where $ V = A_0 \subseteq A_1 \subseteq \dots \subseteq A_p = \emptyset $, such that $ \sum_{v \in V} | \clust (v, \cA, R, G) | = \tilde O( p n^{1 + 1/p}) $.
It also computes for every node~$ v $ the set $ \clust (v, \cA, R, G) $ and for each node $ w \in \clust (v, \cA, R, G) $ the value of $ \dist (v, w, G) $.
\end{theorem}

\subsection{Hop Reduction with Additive Error}\label{sec:hop_reduction_with_additive_error}

Consider the following algorithm for computing a set of edges $ F $.
First, deterministically compute clusters with $ p = \lfloor \sqrt{ \log{n} / \log{(9/\epsilon)}} \rfloor $  priorities (determined by a hierarchy of sets $ \cA = (A_i)_{0 \leq i \leq p} $) up to distance $ R = n^{1/p} \Delta $, where $ \Delta $ is a parameter controlling both the extent of the hop reduction and the quality of the resulting approximation.
Let $ F $ be the set containing an edge for every pair $ (u, v) \in V^2 $ such that $ v \in \clust (u, \cA, R, G) $, and set the weight of such an edge $ (u, v) \in F $ to $ w (u, v, F) = \dist (u, v, G) $, where the distance is returned by the algorithm for computing the clusters.
Procedure~\ref{alg:hop_reduction_additive} presents the pseudocode of this algorithm.

\begin{procedure}
\caption{HopReductionAdditiveError($ G $, $ \Delta $, $ \epsilon $)}
\label{alg:hop_reduction_additive}

\KwIn{Graph $ G = (V, E) $ with positive integer edge weights, $ \Delta \geq 1 $, $ 0 < \epsilon \leq 1 $}
\KwOut{Hop-reducing set of edges $ F \subseteq V^2 $ as specified in \Cref{lem:hop_reduction_additive_error}}

\BlankLine

$ p \gets \left\lfloor \sqrt{\frac{\log{n}}{\log{(9 / \epsilon)}}} \right\rfloor $\;
$ R \gets n^{1/p} \Delta $\;
$ F \gets \emptyset $\;
$ (C (v, \mathcal{A}, R, G), \delta (v, \cdot))_{v \in V} \gets $ \Clusters{$ G $, $ p $, $ R $}\;
\ForEach{$ u \in V $}{
	\ForEach{$ v \in \clust (u) $}{
		$ F \gets F \cup \{ (u, v) \} $\;
		$ w (u, v, F) = \delta (u, v) $\;
	}
}
\KwRet{$ F $}\;
\end{procedure}

\begin{lemma}\label{lem:hop_reduction_additive_error}
Let $ F \subseteq V^2 $ be the set of edges computed by Procedure~\ref{alg:hop_reduction_additive} for a weighted graph $ G = (V, E) $ with positive integer edge weights and parameters $ \Delta \geq 1 $ and $ 0 < \epsilon \leq 1 $.
Then $ F $ has size $ \tilde O (p n^{1 + 1/p}) $, where $ p = \lfloor \sqrt{(\log{n}) / (\log{(9 / \epsilon)})} \rfloor $, and in the graph $ H = G \cup F $, for every pair of nodes $ u $ and $ v $, we have
\begin{equation*}
\dist^{(p+1) \lceil \dist (u, v, G) / \Delta \rceil} (u, v, G) \leq (1 + \epsilon) \dist (u, v, G) + \epsilon n^{1/p} \Delta / (p + 2) \, ,
\end{equation*}
i.e., there is a path $ \pi' $ in $ H $ of weight $ w (\pi', H) \leq (1 + \epsilon) \dist (u, v, G) + \epsilon n^{1/p} \Delta / (p + 2) $ consisting of $ |\pi'| \leq (p+1) \lceil \dist (u, v, G) / \Delta \rceil $ edges.
\end{lemma}
We devote the rest of this section to proving \Cref{lem:hop_reduction_additive_error}.
The bound on the size of~$ F $ immediately follows from \Cref{thm:clusters}.
We analyze the hop-reducing properties of $ F $ by showing the following.
Let $ \pi $ be a shortest path from $ u $ to $ v $ in $ G $.
Then there are a node $ w $ on $ \pi $ and a path $ \pi' $ from $ u $ to $ w $ in $ H = G \cup F $ with the following properties:
\begin{enumerate}[(1)]
\item The distance from $ u $ to $ w $ in $ G $ is at least $ \Delta $.
\item The path $ \pi' $ consists of at most $ p $ edges of $ F $ and at most one edge of $ G $.
\item The ratio between the weight of $ \pi' $ in $ H $ and the distance from $ u $ to $ w $ in $ G $ is at most $ (1 + \epsilon) $ if $ w \neq v $, and if $ w = v $, then the weight of $ \pi' $ in $ H $ is at most $ \beta $ (for some $ \beta $ that we set later).
\end{enumerate}
When we go from $ u $ to $ w $ using the path $ \pi' $ instead of the subpath of $ \pi $ we are using a ``shortcut'' of at most $ p + 1 $ hops that brings us closer to $ v $ by a distance of at least $ \Delta $ at the cost of some approximation.
Conditions~(1) and~(2) guarantee that by repeatedly applying this shortcut we can find a path $ \pi'' $ from $ u $ to $ v $ that has at most $ (p + 1) \lceil \dist (u, v, G) / \Delta \rceil $ hops (as we replace subpaths of $ \pi $ with weight at least $ \Delta $ by paths with at most $ p + 1 $ hops).
Condition~(3) guarantees that the multiplicative error introduced by using the shortcut is at most $ 1 + \epsilon $, except possibly for the last time such a shortcut is used, where we allow an additive error of~$ \beta $.
We show in \Cref{lem:bound_on_additive_error} that we can guarantee a value of $ \beta $ that is bounded by $ \epsilon n^{1/p} / (p + 2) $.
This type of analysis has been used before by Thorup and Zwick~\cite{ThorupZ06} to obtain a spanner for unweighted graphs defined from the partial shortest-path trees of the clusters, but without considering the hop-reduction aspect.
Bernstein~\cite{Bernstein09} also used a similar analysis to obtain a hop set for weighted graphs using clusters with full distance range.
We previously used this type of analysis to obtain a randomized hop set which is based not on clusters, but on a similar notion~\cite{HenzingerKNFOCS14}.

To carry out the analysis as explained above, we define a value $ r_i $ for every $ 0 \leq i \leq p-1 $ as follows:
\begin{align*}
r_0 &= \Delta \\
r_i &= \frac{(4 + 2\epsilon) \sum_{0 \leq j \leq i-1} r_j}{\epsilon} \, .
\end{align*}
The intuition is that a node $ u $ of priority $ i $ tries to take an edge of $ F $ to shortcut the way to $ v $ by at least $ r_i $.
If this fails, it will find an edge in $ F $ going to a node $ v' $ of higher priority.
Thus, to fulfill Condition~(3), $ v' $ has to try to shortcut even more ``aggressively.''
Consequently, the values of $ r_i $ grow exponentially with the priority $ i $.

We have chosen the range of the clusters large enough such that nodes of the highest priority always find the desired shortcut edge in $ F $.
We will show that the additive error incurred by this strategy is at most
\begin{equation*}
\beta = \sum_{0 \leq i \leq p-1} 2 r_i \, .
\end{equation*}
This value can in turn be bounded as follows.

\begin{lemma}\label{lem:bound_on_additive_error}
$ \beta \leq \epsilon n^{1/p} \Delta / (p+2) $.
\end{lemma}

\begin{proof}
We first show that, for all $ 0 \leq i \leq p-1 $, $ \sum_{0 \leq j \leq i} r_j \leq 7^i \Delta / \epsilon^i $.
The proof is by induction on $ i $.
For $ i = 0 $ we have $ r_0 = \Delta = 7^0 \Delta $ by the definition of $ r_0 $ and for $ i \geq 1 $ we use the inequalities $ \epsilon > 0 $ and $ \epsilon \leq 1 $ and the induction hypothesis as follows:
\begin{align*}
\sum_{0 \leq j \leq i} r_j &= \sum_{0 \leq j \leq i-1} r_j + r_i = \sum_{0 \leq j \leq i-1} r_j + \frac{(4 + 2 \epsilon) \sum_{0 \leq j \leq i-1} r_j}{\epsilon} \\
 &\leq \frac{\sum_{0 \leq j \leq i-1} r_j}{\epsilon} + \frac{6 \sum_{0 \leq j \leq i-1} r_j}{\epsilon}
\leq \frac{7 \sum_{0 \leq j \leq i-1} r_j}{\epsilon} \leq \frac{7 \cdot 7^{i-1} \Delta}{\epsilon \cdot \epsilon^{i-1}}  \leq \frac{7^i \Delta}{\epsilon^i} \, .
\end{align*}
Using this inequality and the fact that $ (2p+7) 7^{p-1} \leq 9^p $ for all $ p \geq 0 $, we get
\begin{equation*}
\frac{(p+2) \beta}{\epsilon} = \frac{(p+2) \sum_{0 \leq j \leq p-1} 2 r_j}{\epsilon} \leq \frac{(2p+4) 7^{p-1} \Delta}{\epsilon^p} \leq \frac{9^p \Delta}{\epsilon^p} \leq n^{1/p} \Delta
\end{equation*}
The last inequality holds as by our choice of $ p = \lfloor \sqrt{ \log{n} / \log{(9/\epsilon)}} \rfloor $,
\begin{align*}
(9 / \epsilon)^p = 2^{p \cdot \log{(9 / \epsilon)}} 
 \leq 2^{\frac{\sqrt{\log{n}}}{\sqrt{\log{(9 / \epsilon)}}} \cdot \log{(9 / \epsilon)}} &= 2^{\sqrt{\log{n}} \cdot \sqrt{\log{(9 / \epsilon)}}} \\
 &= 2^{\frac{\sqrt{\log{(9 / \epsilon)}}}{\sqrt{\log{n}}} \cdot \log{n}} = 2^{\log{n} \cdot (1/p)} = n^{1/p} \, .
\end{align*}
\end{proof}

In the following we fix some values of $ \epsilon $ and $ \Delta $ and let $ F $ denote the set of edges computed by Procedure~\ref{alg:hop_reduction_additive}.
We now show that $ F $ has a certain structural property before we carry out the hop-reduction proof.

\begin{lemma}\label{lem:hop_set_structural_property}
Let $ u $ and $ v $ be nodes such that $ u \in A_i \setminus A_{i+1} $ (i.e., $ u $ has priority $ i $) and $ \dist (u, v, G) \leq r_i $.
Either % Note: it might also happen that both cases apply
\begin{enumerate}[(1)]
\item $ F $ contains an edge $ (u, v) $ of weight $ w (u, v, F) = \dist (u, v, G) $ or
\item $ F $ contains an edge $ (u, v') $ to a node $ v' \in A_{i+1} $ of priority $ j \geq i + 1 $ of weight $ w (u, v', F) \leq 2 r_i $.
\end{enumerate}
\end{lemma}

\begin{proof}
First consider the case $ v \in \clust (u, \cA, R, G) $.
Then $ F $ contains the edge $ (u, v) $ of weight $ w (u, v, F) = \dist (u, v, G) $.

Now consider the case $ v \notin \clust (u, \cA, R, G) $.
Note that by the definition of $ \beta $ we have $ r_i \leq \beta/2 < \beta $, and by \Cref{lem:bound_on_additive_error} we have $ \beta \leq n^{1/p} \Delta $ (where $ p $ is the parameter controlling the construction of clusters).
As the algorithm sets $ R = n^{1/p} \Delta $, we have $ r_i \leq R $ by \Cref{lem:bound_on_additive_error} and thus $ \dist (u, v, G) \leq R $.
From the definition of $ \clust (u, \cA, R, G) $ it now follows that $ \dist (v, u, G) \geq \dist (v, A_{i+1}, G) $.
Thus there exists some node $ v_1 \in A_{i+1} $ of priority $ p_1 \geq i+1 $ such that $ \dist (v, v_1, G) \leq \dist (u, v, G) $.
By the triangle inequality we get $ \dist (u, v_1, G) \leq \dist (u, v, G) + \dist (v, v_1, G) \leq 2 \dist (u, v, G) \leq 2 r_i \leq R $.
If $ u \in \clust (v_1, \cA, R, G) $, then we are done as $ F $ contains the edge $ (u, v_1) $ of weight $ w (u, v_1, F) = \dist (u, v_1, G) \leq 2 r_i $.
Otherwise it follows from the definition of $ \clust (v_1, \cA, R, G) $ that there is some node $ v_2 \in A_{p_1+1} $ of priority $ p_2 \geq p_1 + 1 \geq i + 1 $ such that $ \dist (u, v_2, G) \leq \dist (u, v_1, G) \leq 2 r_i \leq R $.
By repeating the argument above we want to find some node $ v_j \in A_{i+1} $ of priority $ p_j \geq i + 1 $ such that $ \dist (u, v_j, G) \leq \dist (u, v, G) \leq 2 r_i $.
As for every node $ v' \in A_{p-1} $ of priority $ p-1 $, $ \clust (v', \cA, R, G) $ contains all nodes that are at distance at most~$ R $ from~$ v' $ in~$ G $, this repeated argument stops eventually and we find such a node.
\end{proof}

To finish the proof of \Cref{lem:hop_reduction_additive_error} we show in the next lemma that $ F $ has the properties we demanded, i.e., in the shortcut graph $ H $ which consists of $ G $ and the additional edges of $ F $, we can approximate shortest paths using a reduced number of hops.

\begin{lemma}\label{pro:hop_set_layers_additive_approximation}
For every pair of nodes $ u, v \in V $ such that $ \dist (u, v, G) < \infty $, the graph $ H = G \cup F $ contains a path $ \pi' $ from $ u $ to $ v $ of weight $ w (\pi', H) \leq (1 + \epsilon) \dist (u, v, G) + \beta $ consisting of $ | \pi' | \leq (p+1) \lceil \dist (u, v, G) / \Delta \rceil $ edges.
\end{lemma}

\begin{proof}
The proof is by induction on the distance from $ u $ to $ v $ in $ G $.
The claim is trivially true for the base case $ \dist (u, v, G) = 0 $ in which $ u = v $.
Thus, we only need to consider the induction step in which $ \dist (u, v, G) \geq 1 $.

Let $ \pi $ denote the shortest path from $ u $ to $ v $ in $ G $.
We now define a sequence of nodes $ u_0, u_1, \ldots, u_l $, where $ l \leq p - 1 $.
For every $ 0 \leq j \leq l $, we denote by $ p_j $ the priority of $ u_j $.
We set $ u_0 = u $, and, given~$ u_j $, we define $ u_{j+1} $ as follows.
Let $ x $ be the node on $ \pi $ closest to~$ v $ that is at distance at most~$ r_{p_j} $ from $ u_j $ in $ G $ (this node might be $ v $ itself).
If $ F $ contains the edge $ (u_j, x) $ we stop (and set $ l = j $).
Otherwise we know by \Cref{lem:hop_set_structural_property} that $ F $ (and therefore also $ H $) contains an edge $ (u_j, u') $ to a node $ u' $ of priority at least $ p_j + 1 $.
In that case we set $ u_{j+1} = u' $.
We know further by \Cref{lem:hop_set_structural_property} that $ \dist (u_j, u_{j+1}, G) \leq 2 r_{p_{j+1}-1} $.
Note that the definition of this sequence $ u_0, u_1, \ldots, u_l $ ensures that $ l \leq p - 1 $ as the priority strictly increases with each node $ u_j $.
Having defined the sequence, we denote by $ x $ the node on $ \pi $ closest to~$ v $ that is at distance at most~$ r_{p_l} $ from $ u_l $ in $ G $ (again, this node might be $ v $ itself).
The definition of $ u_l $ guarantees that $ F $ (and thus H) contains the edge $ (u_l, x) $.
Figure~\ref{fig:illustration_of_path} illustrates the definition of this sequence.

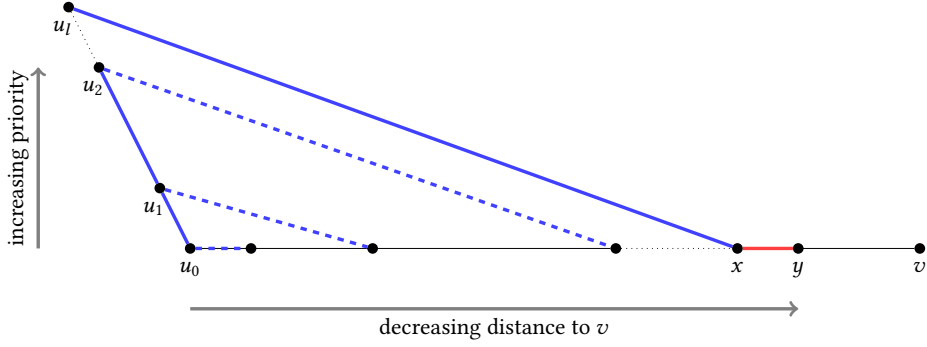
\begin{figure}[htbp!]
\centering
\scalebox{0.8}{
\begin{tikzpicture}
\tikzstyle{vertex}=[circle,fill=black,minimum size=5pt,inner sep=0pt,outer sep=0pt]
\tikzstyle{shortest-path} = [draw]
\tikzstyle{graph-edge} = [draw,ultra thick,-,color=red!75]
\tikzstyle{hop-edge} = [draw,ultra thick,-,color=blue!75]
\tikzstyle{hop-edge-missing} = [draw,ultra thick,dashed,color=blue!75]

\node[vertex,label=below:$u_0$] (u0) at (0, 0) {};
\node[vertex,label=below:$u_1~~$] (u1) at (-0.5, 1) {};
\node[vertex,label=below:$u_2~~$] (u2) at (-1.5, 3) {};
\node[vertex,label=below:$u_l~~$] (ul) at (-2, 4) {};
\node[vertex] (x0) at (1, 0) {};
\node[vertex] (x1) at (3, 0) {};
\node[vertex] (x2) at (7, 0) {};
\node[vertex,label=below:$x$] (w) at (9, 0) {};
\node[vertex,label=below:$y$] (w2) at (10, 0) {};
\node[vertex,label=below:$v$] (v) at (12, 0) {};

\path[shortest-path] (u0) -- (x2);
\path[draw,dotted] (x2) -- (w);
\path[shortest-path] (w) -- (v);

\path[hop-edge] (u0) -- (u1);
\path[hop-edge] (u1) -- (u2);
\path[hop-edge] (ul) -- (w);
\path[hop-edge-missing] (u0) -- (x0);
\path[hop-edge-missing] (u1) -- (x1);
\path[hop-edge-missing] (u2) -- (x2);
\path[graph-edge] (w) -- (w2);

\path[draw,dotted] (u2) -- (ul);

\path[draw,ultra thick,color=black!50,->] (0,-1) -- node[below,color=black]{decreasing distance to $ v $} (10,-1) ;
\path[draw,ultra thick,color=black!50,->] (-2.5,0) -- node[above,color=black,rotate=90]{increasing priority} (-2.5,3) ;
\end{tikzpicture}
}
\caption{Schematic illustration of the definition of the sequence of nodes $ u_0, u_1, \ldots, u_l, x $.
The bottom line represents the shortest path $ \pi $ from $ u $ to $ v $.
The thick blue edges are the edges from~$ F $ used to shorten the distance to~$ v $. The dashed blue edges are not contained in $ F $ and imply the existence of edges to nodes of increasing priority.
The node $ y $ is the successor of $ x $ on $ \pi $, and the thick red edge $ (x, y) $ from $ G $ is the last edge on the shortcut path from $ u_0 $ to $ y $.
The dotted lines indicate repetitions that are omitted in the picture.}\label{fig:illustration_of_path}
\end{figure}

First consider the case that $ x = v $.
Let $ \pi' $ denote the path $ \langle u_0, \ldots, u_l, x \rangle $.
This path has at most $ p $ hops and since $ \dist (u, v, G) \geq 1 $ we trivially have $ p \leq (p+1) \lceil \dist (u, v, G) / \Delta \rceil $.
Furthermore we can bound the weight of $ \pi' $ as follows:
\begin{align*}
w (\pi', H) &= \sum_{0 \leq j \leq l-1} w (u_j, u_{j+1}, H) + w (u_l, x, H) \\
 &\leq \sum_{0 \leq j \leq l-1} \dist (u_j, u_{j+1}, G) + \dist (u_l, x, G) \\
 &\leq \sum_{0 \leq j \leq l-1} 2 r_{p_{j+1}-1} + r_{p_l} \\
 &\leq \sum_{0 \leq j \leq p-1} 2 r_j \\
 &= \beta \leq (1 + \epsilon) \dist (u, v, G) + \beta \, .
\end{align*}

Now consider the case $ x \neq v $.
Let $ y $ be the neighbor of~$ x $ on~$ \pi $ (which in $ G $ is closer to~$ v $ than $ x $ is).
We will define the path $ \pi' $ from $ u $ to $ v $ as the concatenation of two paths $ \pi_1 $ and $ \pi_2 $.
Let $ \pi_1 $ be the path $ \langle u_0, \ldots, u_l, x, y \rangle $.
We will define the path $ \pi_2 $ from $ y $ to $ v $ later.
Note that $ \pi_1 $ consists of $ |p_1| \leq p + 1 $ hops.
We will now show that
\begin{equation}
w (\pi_1, H) \leq (1 + \epsilon) \dist (u, y, G) \, . \label{eq:bound_on_weight_of_path}
\end{equation}
In order to get this bound we will need some auxiliary inequalities.
By \Cref{lem:hop_set_structural_property} we have, for all $ 0 \leq j \leq l-1 $,
\begin{equation}
\dist (u_j, u_{j+1}, G) \leq 2 r_{p_{j+1}-1} \label{eq:priority_difference}
\end{equation}
and by the definition of $ r_{p_l} $ we have
\begin{equation}
\epsilon r_{p_l} = (4 + 2 \epsilon) \sum_{0 \leq j \leq p_l-1} r_j \, . \label{eq:radius_sum}
\end{equation}
Remember that $ x $ is the node on the path~$ \pi $ closest to~$ v $ that is at distance at most~$ r_{p_l} $ from $ u_l $ in~$ G $.
Since the neighbor $ y $ of $ x $ is closer to $ v $ than $ x $ is, this definition of~$ x $ guarantees that $ \dist (u_l, y, G) > r_{p_j} $.
As $ \dist (u_l, y, G) \leq \dist (u_l, x, G) + \dist (x, y, G) $ by the triangle inequality, we have
\begin{equation}
\dist (u_l, x, G) + \dist (x, y, G) > r_{p_j} \, . \label{eq:distance bounding inequality}
\end{equation}
By the triangle inequality we also have
\begin{equation*}
\dist (u_l, x, G) \leq \sum_{0 \leq j \leq l-1} \dist (u_j, u_{j+1}, G) + \dist (u, x, G)
\end{equation*}
and thus
\begin{equation}
\dist (u_l, x, G) - \sum_{0 \leq j \leq l-1} \dist (u_j, u_{j+1}, G) \leq \dist (u, x, G) \, . \label{eq:triangle_inequality}
\end{equation}

We now obtain~\eqref{eq:bound_on_weight_of_path} as follows:
\begin{align*}
w (\pi_1, H) &= \sum_{0 \leq j \leq l-1} w (u_j, u_{j+1}, H) + w (u_l, x, H) + w (x, y, H)  \\
 &= \sum_{0 \leq j \leq l-1} \dist (u_j, u_{j+1}, G) + \dist (u_l, x, G) + \dist (x, y, G) \\
 &= (2 + \epsilon) \sum_{0 \leq j \leq l-1} \dist (u_j, u_{j+1}, G) + \dist (u_l, x, G) + \dist (x, y, G) \\
 & \hspace{1em} - (1 + \epsilon) \sum_{0 \leq j \leq l-1} \dist (u_j, u_{j+1}, G) \\
 &\stackrel{\mathclap{\eqref{eq:priority_difference}}}{\leq} (2 + \epsilon) \sum_{0 \leq j \leq l-1} 2 r_{p_{j+1}-1} + \dist (u_l, x, G) + \dist (x, y, G) \\
 &\hspace{1em} - (1 + \epsilon) \sum_{0 \leq j \leq l-1} \dist (u_j, u_{j+1}, G) \\
 &\leq (2 + \epsilon) \sum_{0 \leq j \leq p_l-1} 2 r_j + \dist (u_l, x, G) + \dist (x, y, G) \\
 &\hspace{1em}- (1 + \epsilon) \sum_{0 \leq j \leq l-1} \dist (u_j, u_{j+1}, G) \\
 &\stackrel{\mathclap{\eqref{eq:radius_sum}}}{=} \epsilon r_{p_l} + \dist (u_l, x, G) + \dist (x, y, G) - (1 + \epsilon) \sum_{0 \leq j \leq l-1} \dist (u_j, u_{j+1}, G) \\
 &\stackrel{\mathclap{\eqref{eq:distance bounding inequality}}}{<} \epsilon (\dist (u_l, x, G) + \dist (x, y, G)) + \dist (u_l, x, G) + \dist (x, y, G) \\
 &\hspace{1em} - (1 + \epsilon) \sum_{0 \leq j \leq l-1} \dist (u_j, u_{j+1}, G) \\
 &= (1 + \epsilon) \left( \dist (u_l, x, G) - \sum_{0 \leq j \leq l-1} \dist (u_j, u_{j+1}, G) \right) + (1 + \epsilon) \dist (x, y, G) \\
 &\stackrel{\mathclap{\eqref{eq:triangle_inequality}}}{\leq} (1 + \epsilon) \dist (u, x, G) + (1 + \epsilon) \dist (x, y, G) \\
 &\leq (1 + \epsilon) (\dist (u, x, G) + \dist (x, y, G)) = (1 + \epsilon) \dist (u, y, G)
\end{align*}

Note that $ \dist (y, v, G) < \dist (u, v, G) $.
Therefore we may apply the induction hypothesis on $ y $ and get that the graph $ H $ contains a path $ \pi_2 $ from $ y $ to $ v $ of weight $ w (\pi_2, H) \leq (1 + \epsilon) \dist (y, v, G) + \beta $ that has $ | \pi_2 | \leq (p+1) \lceil \dist (y, v, G) / \Delta \rceil $ hops.
Let $ \pi' $ denote the concatenation of $ \pi_1 $ and $ \pi_2 $.
Then $ \pi' $ is a path from $ u $ to $ v $ in $ H $ of weight
\begin{align*}
w (\pi', H) &= w (\pi_1, H) + w (\pi_2, H) \\
 &\leq (1 + \epsilon) \dist (u, y, G) + (1 + \epsilon) \dist (y, v, G) + \beta \\
 &= (1 + \epsilon) (\dist (u, y, G) + \dist (y, v, G)) + \beta \\
 &= (1 + \epsilon) \dist (u, v, G) + \beta \, .
\end{align*}

It remains to bound the number of hops of $ \pi' $.
To get the desired bound we first show that $ \dist (u, y, G) \geq \Delta $.
By the triangle inequality we have
\begin{equation*}
\dist (u_l, y, G) \leq \sum_{0 \leq j \leq l-1} \dist (u_j, u_{j+1}, G) + \dist (u, y, G) \, .
\end{equation*}
As argued above, we have $ \dist (u_l, y, G) > r_{p_j} $ and
\begin{equation*}
\sum_{0 \leq j \leq l-1} \dist (u_j, u_{j+1}, G) \leq \sum_{0 \leq j \leq p_l - 1} 2 r_j \, .
\end{equation*}
By the definition of $ r_{p_l} $ we therefore get
\begin{align*}
\dist (u, y, G) &\geq \dist (u_l, y, G) - \sum_{0 \leq j \leq l-1} \dist (u_j, u_{j+1}, G) \\
 &\geq r_{p_l} - \sum_{0 \leq j \leq p_l - 1} 2 r_j = (4 / \epsilon) \sum_{0 \leq j \leq p_l - 1} r_j \geq r_0 = \Delta \, .
\end{align*}
Now that we know that $ \dist (u, y, G) \geq \Delta $, or equivalently $ \dist (u, y, G) / \Delta \geq 1 $, we get the following for counting the number of hops of $ \pi' $ by adding the number of hops of $ \pi_1 $ to the number of hops of $ \pi_2 $:
\begin{align*}
| \pi' | = | \pi_1 | + | \pi_2 | &\leq
p+1 + (p+1) \lceil \dist (y, v, G) / \Delta \rceil \\
 &= (p+1) (1 + \lceil \dist (y, v, G) / \Delta \rceil) \\
 &= (p+1) \lceil 1 + \dist (y, v, G) / \Delta \rceil \\
 &\leq (p+1) \lceil \dist (u, y, G) / \Delta + \dist (y, v, G) / \Delta \rceil \\
 &= (p+1) \lceil (\dist (u, y, G) + \dist (y, v, G)) / \Delta \rceil \\
 &= (p+1) \lceil \dist (u, v, G) / \Delta \rceil \, .
\end{align*}
Thus, $ \pi' $ has the desired number of edges.
\end{proof}

\subsection{Hop Reduction without Additive Error}\label{sec:hop_reduction_without_additive_error}

Consider a shortest path $ \pi $ from $ u $ to $ v $ with $ h $ hops and weight $ R \geq \Delta $.
With the hop reduction of Procedure~\ref{alg:hop_reduction_additive} we can compute a set of edges $ F $ such that in $ G \cup F $ we find a path from $ u $ to $ v $ with $ \tilde O (R / \Delta) $ hops of weight approximately $ R $ (where we incur an additive error of roughly $ \Delta n^{o(1)} $).
We now use the weight-rounding technique of \Cref{thm:property of weight rounding} and repeatedly apply this algorithm to obtain a set of edges $ F $ such that in $ G \cup F $ there is a path from $ u $ to $ v $ with $ O (h / \Delta) $ hops and weight approximately~$ R $.
As in general $ R $ can only be upper-bounded by $ n W $ (where $ W $ is the maximum edge weight of $ G $) and $ h $ can be upper-bounded by $ n $, the second type of hop reduction seems more desirable.
Additionally, if $ h $ is sufficiently larger than $ \Delta $, then the additive error inherent in the hop reduction of Procedure~\ref{alg:hop_reduction_additive} can be counted as a small multiplicative error.\footnote{Note that for smaller values of $ h $, $ \pi $ itself has a small enough number of hops and thus there is no need to find a path in $ G \cup F $ with a small number of hops and weight approximately $ R $.}

The second hop-reduction algorithm roughly works as follows.
For every possible distance range of the form $ 2^j \ldots 2^{j+1} $ we scale down the edge weights of $ G $ by a certain factor and run the algorithm of Procedure~\ref{alg:hop_reduction_additive} on the modified graph $ \hat{G}_j $ to compute a set of edges $ \hat{F}_j $.
We then simply return the union of all these edge sets (with the weights scaled back to normal again).
Procedure~\ref{alg:hop_reduction} shows the pseudocode of this algorithm.

\begin{procedure}
\caption{HopReduction($ G $, $ \Delta $, $ h $, $ \epsilon $, $ W $)}
\label{alg:hop_reduction}

\KwIn{Weighted graph $ G = (V, E) $ with positive integer edge weights in $ \{ 1, \ldots, W \} $, $ \Delta \geq 1 $, $ h \geq 1 $, $ 0 < \epsilon \leq 1 $}
\KwOut{Hop-reducing set of edges $ F \subseteq V^2 $ as specified in \Cref{lem:hop_reduction}}

\BlankLine

$ \epsilon' \gets \frac{\epsilon}{7} $\;
$ \Delta' \gets \frac{4 \Delta}{\epsilon'} $\;
$ F \gets \emptyset $\;
\For{$ j = 0 $ \KwTo $ \lfloor \log(nW) \rfloor$}{
	$ \hat{G}_j \gets (V, E) $\;
	$ \rho_j \gets \frac{\epsilon' 2^j}{h} $\;
	\lForEach{$ (u, v) \in E $}{
		$ w (u, v, \hat{G}_j) \gets \left\lceil \frac{w (u, v, G)}{\rho_j} \right\rceil $
	}
	$ \hat{F}_j \gets $ \HopReductionAdditiveError{$\hat{G}_j$, $\Delta'$, $\epsilon'$}\;
	\ForEach{$ (u, v) \in \hat{F}_j $}{
		$ F \gets F \cup \{ (u, v) \} $\;
		$ w (u, v, F) \gets \min (w (u, v, \hat{F}_j) \cdot \rho_j, w (u, v, F)) $\; 
	}
}
\KwRet{$ F $}\;
\end{procedure}

\begin{lemma}\label{lem:hop_reduction}
Let $ F \subseteq V^2 $ be the set of edges computed by Procedure~\ref{alg:hop_reduction} for a weighted graph $ G = (V, E) $ with positive integer edge weights in $ \{ 1, \ldots, W \} $ and parameters $ \Delta \geq 1 $, $ h \geq 1 $, and $ 0 < \epsilon \leq 1 $.
Then $ F $ has size $ \tilde O (p n^{1 + 1/p} \log{W}) $, where $ p = \lfloor \sqrt{(\log{n}) / (\log{(63 / \epsilon)})} \rfloor $,
and if $ h \geq n^{1/p} \Delta / (p + 2) $, then in the graph $ H = G \cup F $ we have, for every pair of nodes $ u $ and $ v $,
\begin{equation*}
\dist^{(p+2) h / \Delta} (u, v, H) \leq (1 + \epsilon) \dist^h (u, v, G) \, .
\end{equation*}
\end{lemma}

\begin{proof}
Let $ u $ and $ v $ be a pair of nodes, and set $ j = \lfloor \log{\dist^h (u, v, G)} \rfloor $, i.e., $ 2^j \leq \dist^h (u, v, G) \leq 2^{j + 1} $.
Let $ \pi $ be a shortest $ \leq h $ hop path in $ G $; i.e., $ \pi $ has weight $ w (\pi, G) = \dist^h (u, v, G) $, and $ \pi $ consists of $ | \pi | \leq h $ hops.
The algorithm sets $ \epsilon' = \epsilon / 7 $ and uses a graph $ \hat{G}_j $ which has the same nodes and edges as~$ G $, but in which every edge weight is first scaled down by a factor of $ \rho_j = \epsilon' 2^j / h $ and then rounded up to the next integer.
By \Cref{thm:property of weight rounding} we have $ \dist (u, v, \hat{G}_j) \cdot \rho_j \leq (1 + \epsilon') \dist^h (u, v, G) $ and $ \dist (u, v, \hat{G}_j) \leq (2 + 2/\epsilon') h \leq 4 h / \epsilon' $.

Consider the set of edges $ \hat{F}_j $ computed in Procedure~\ref{alg:hop_reduction} (such a set does indeed exist because $ \dist^h (u, v, G) \leq nW $).
By \Cref{lem:hop_reduction_additive_error}, there is a path $ \pi' $ in $ \hat{H}_j = \hat{G}_j \cup \hat{F}_j $ of weight at most $ (1 + \epsilon') \dist (u, v, \hat{G}_j) + \epsilon' n^{1/p} \Delta' / (p + 2) $ and with at most $ | \pi' | \leq (p + 1) \lceil \dist (u, v, \hat{G}_j) / \Delta' \rceil $ hops.
Since we have $ \dist (u, v, \hat{G}_j) \leq 4 h / \epsilon' $ and the algorithm sets $ \Delta' = 4 \Delta / \epsilon' $, we get
\begin{align*}
| \pi' | &\leq (p + 1) \cdot \left\lceil \frac{ \dist (u, v, \hat{G}_j)}{\Delta'} \right\rceil \\
 &\leq (p + 1) \cdot \left\lceil \frac{4 h}{\epsilon' \Delta'} \right\rceil \\
 &= (p + 1) \cdot \left\lceil \frac{h}{\Delta} \right\rceil \\
 &\leq (p + 1) \left( \frac{h}{\Delta} + 1 \right) \\
 &= \frac{(p + 1) h}{\Delta} + (p + 1) \\
 &\leq \frac{(p + 1) h}{\Delta} + 4^p \leq \frac{(p + 1) h}{\Delta} + n^{1/p} \leq \frac{(p + 1) h}{\Delta} + \frac{h}{\Delta} = \frac{(p + 2) h}{\Delta} \, .
\end{align*}

The algorithm ``scales back'' the edge weights of $ \hat{F}_j $ when adding them to $ F $ and thus $ w (u, v, F) \leq w (u, v, \hat{F}_j) \cdot \rho_j $.
We now argue that $ \dist^{(p+1) \lceil h / \Delta \rceil} (u, v, H)  \leq (1 + \epsilon') \dist^h (u, v, G) $ by bounding the weight of $ \pi' $ in $ H = G \cup F $.
For every edge $ (u, v) $ of $ \pi' $ we have $ w (u, v, H) \leq w (u, v, F) \leq w (u, v, \hat{F}_j) \cdot \rho_j $ if $ (u, v) \in \hat{F}_j $ and $ w (u, v, H) \leq w (u, v, G) \leq w (u, v, \hat{G}_j) \cdot \rho_j $ otherwise.
Thus, $ w (\pi', H) \leq w (\pi', \hat{H}_j) \cdot \rho_j $ and together with the assumption $ h \geq n^{1/p} \Delta / (p + 2) $ we get
\begin{align*}
\dist^{(p+2) h / \Delta} (u, v, H) \leq
w (\pi', H) &\leq w (\pi', \hat{H}_j) \cdot \rho_j \\
 &\leq \left( (1 + \epsilon') \dist (u, v, \hat{G}_j) + \frac{\epsilon' n^{1/p} \Delta'}{p + 2} \right) \cdot \rho_j \\
 &= (1 + \epsilon') \dist (u, v, \hat{G}_j) \cdot \rho_j + \frac{\epsilon' n^{1/p} \Delta' \rho_j}{p + 2} \\
 &= (1 + \epsilon') \dist (u, v, \hat{G}_j) \cdot \rho_j + \frac{4 \epsilon' 2^j n^{1/p} \Delta}{h (p + 2)} \\
 &\leq (1 + \epsilon') \dist (u, v, \hat{G}_j) \cdot \rho_j + 4 \epsilon' 2^j \\
 &\leq (1 + \epsilon')^2 \dist^h (u, v, G) + 4 \epsilon' \dist^h (u, v, G) \\
 &\leq (1 + 7 \epsilon') \dist^h (u, v, G) \\
 &= (1 + \epsilon) \dist^h (u, v, G) \, .
\end{align*}
\end{proof}

\subsection{Computing the Hop Set}\label{sec:computing_hop_set}

We finally explain how to repeatedly use the hop reduction of Procedure~\ref{alg:hop_reduction} to obtain an $ (n^{o(1)}, o(1)) $-hop set.
Procedure~\ref{alg:hop_reduction} computes a set of edges $ F $ that reduces the number of hops needed to approximate the distance between any pair of nodes by a factor of $ 1 / \Delta $ (where $ \Delta $ is a parameter).
Intuitively we would now like to use a large value of $ \Delta $ to compute a hop set.
However, we want to avoid large values of $ \Delta $ for two reasons.
The first reason is that $ F $ reduces the number of hops only if the shortest path has $ h \geq \Delta n^{o(1)} $ hops.
Thus, for shortest paths that already have $ h < \Delta $ hops, the hop reduction is not effective.
The second reason is efficiency.
The algorithm requires us to compute clusters for distances up to $ \Delta n^{o(1)} $, and in the models of computation we consider later on, we do not know how to do this fast enough for our purposes.

We therefore use the following iterative approach in which we repeatedly apply the hop reduction of Procedure~\ref{alg:hop_reduction} with $ \Delta = (p+2) n^{1/p} = n^{o(1)} $.
We first compute a set of edges $ F_1 $ that reduces the number of hops in $ G $ by a factor of $ 1 / \Delta $.
We then add all these edges to $ G $ and consider the graph $ H_1 = G \cup F_1 $.
We apply the algorithm again on $ H_1 $ to compute a set of edges $ F_2 $ that reduces the number of hops in $ H_1 $ by a factor of $ 1 / \Delta $.
Now observe that the set of edges $ F_1 \cup F_2 $ reduces the number of hops in $ G $ by a factor of $ 1 / \Delta^2 $.
We show that by repeating this process $ p = \Theta (\sqrt{\log{n} / \log{(\sqrt{\log n} / \epsilon)}}) $ times we can compute a set $ F $ that reduces the number of hops to $ n^{1/p} $.
Procedure~\ref{alg:hop_set} shows the pseudocode of this algorithm.

\begin{procedure}
\caption{HopSet($ G $, $ \epsilon $, $ W $)}
\label{alg:hop_set}

\KwIn{Weighted graph $ G = (V, E) $ with positive integer edge weights in $ \{ 1, \ldots, W \} $, $ 0 < \epsilon \leq 1 $}
\KwOut{ $ (n^{1/p}, \epsilon) $-hop set $ F \subseteq V^2 $ as specified in \Cref{thm:hop_set}}

\BlankLine

$ \epsilon' \gets \frac{\epsilon}{2 \sqrt{\log{n}}} $\;
$ W' \gets (1 + \epsilon) n W $\;
$ p \gets \left\lfloor \sqrt{\frac{\log{n}}{\log{(63 / \epsilon')}}} \right\rfloor $\;
$ \Delta \gets (p+2) n^{1/p} $\;
$ F \gets \emptyset $\;
$ H_0 \gets G $\;
\For{$ i = 0 $ \KwTo $ p - 1 $}{
	$ h_i \gets n^{1 - i/p} $\;
	$ F_{i+1} \gets $ \HopReduction{$H_i$, $ \Delta $, $h_i$, $ \epsilon' $, $ W' $}\;
	$ F \gets F \cup F_{i+1} $\;
	$ H_{i+1} \gets H_i \cup F_{i+1} $\;
}
\KwRet{$ F $}\;
\end{procedure}

\begin{theorem}\label{thm:hop_set}
Let $ F \subseteq V^2 $ be the set of edges computed by Procedure~\ref{alg:hop_set} for a weighted graph $ G = (V, E) $ with positive integer edge weights in $ \{ 1, \ldots, W \} $ and a parameter $ 0 < \epsilon \leq 1 $.
Then $ F $ is an $ (n^{1/p}, \epsilon) $-hop set of size $ \tilde O (p^2 n^{1 + 1/p} \log{W}) $, where $ p = \lfloor \sqrt{(\log{n}) / (\log{(126 \sqrt{\log{n}} / \epsilon)})} \rfloor $.
\end{theorem}

\begin{proof}
The algorithm sets $ \epsilon' = \epsilon / (2 \sqrt{\log{n}}) $ and $ p = \lfloor \sqrt{(\log{n}) / (\log{(63 / \epsilon')})} \rfloor $ and uses a parameter $ h_i = n^{1-i/p} $ for each graph $ H_i $.
For every $ 0 \leq i \leq p-2 $ we set $ h_i = n^{1-i/p} \geq n^{2/p} = n^{1/p} \Delta / (p + 2) $, and thus, by \Cref{lem:hop_reduction}, for every pair of nodes $ u $ and $ v $ we have
\begin{equation*}
\dist^{h_{i+1}} (u, v, H_{i+1}) = \dist^{h_i / n^{1/p}} (u, v, H_{i+1}) = \dist^{(p+2) h_i / \Delta} (u, v, H_{i+1}) \leq (1 + \epsilon') \dist^{h_i} (u, v, H_i) \, .
\end{equation*}
By iterating this argument we get
\begin{equation*}
\dist^{h_i} (u, v, H_i) \leq (1 + \epsilon')^i \dist^{h_0} (u, v, H_0) = (1 + \epsilon')^i \dist^n (u, v, G) = (1 + \epsilon')^i \dist (u, v, G)
\end{equation*}
for every $ 1 \leq i \leq p-1 $, and now in particular for $ i = p-1 $ we have
\begin{equation*}
\dist^{n^{1/p}} (u, v, H_{p-1}) = \dist^{h_{p-1}} (u, v, H_{p-1}) \leq (1 + \epsilon')^{(p-1)} \dist (u, v, G) \, .
\end{equation*}
Finally, since $ p-1 \leq \sqrt{\log{n}} $, we have, by the inequality $ (1 + x/y)^{y/2} \leq 1 + x $ for all $ 0 \leq x \leq 1 $ and $ y > 0 $\footnote{This inequality follows from the three well-known inequalities $ (1 + 1/z)^z \leq e $ (for all $ z > 0 $), $ e^z \leq 1 / (1 - z) $ (for all $ z < 1 $), and $ 1 / (1 - z) \leq 1 + 2 z $ (for all $ 0 \leq z \leq 1/2 $), where $ e $ is Euler's constant.},
\begin{equation*}
(1 + \epsilon')^{p-1} = \left( 1 + \frac{\epsilon}{2 \sqrt{\log{n}}} \right)^{p-1} \leq \left( 1 + \frac{\epsilon}{2 \sqrt{\log{n}}} \right)^{\sqrt{\log{n}}} \leq 1 + \epsilon \, .
\end{equation*}
As $ H_{p-1} = G \cup F $, it follows that $ \dist^{n^{1/p}} (u, v, G \cup F) \leq (1 + \epsilon) \dist (u, v, G) $ and thus $ F = \bigcup_{1 \leq i \leq p-1} F_i $ is an $ (n^{1/p}, \epsilon) $-hop set.
\end{proof}

The main computational cost for constructing the hop set comes from computing the clusters in Procedure~\ref{alg:hop_reduction_additive}, which is used as a subroutine repeatedly.
Observe that in total it will perform $ O (p \log{(nW)}) $ calls to Procedure~\ref{alg:clusters} to compute clusters, each with $ p = \Theta( \sqrt{(\log{n}) / (\log{(\sqrt{\log{n}} / \epsilon)})} ) $ priorities and distance range $ R = O (p \sqrt{\log{n}} n^{2/p} / \epsilon) $ on a weighted graph of size $ \tilde O (m + p^2 n^{1 + 1/p} \log{W}) $.
Note that if $ 1/\epsilon \leq \polylog{n} $, then $ n^{1/p} = n^{o(1)} $.
Thus, Procedure~\ref{alg:hop_set} will then compute an $ (n^{o(1)}, o(1)) $-hop set of size $ O (n^{1 + o(1)} \log{W}) $ and it will perform $ \tilde O (\log{W}) $ cluster computations with $ p = \Theta (\sqrt{\log{n} / \log{\log{n}}}) $ priorities up to distance range $ O (n^{o(1)}) $ on graphs of size $ O (m^{1 + o(1)} \log{W}) $ each, where $ m $ is the number of edges of the input graph.

\section{Distributed Approximate Single-Source Shortest Paths Algorithm on Networks with Arbitrary Topology}\label{sec:dist algo}

In this section we describe a deterministic distributed algorithm for computing distances from a source node~$s$.
It consists of two parts. The first part is constructing a suitable {\em overlay network}. A {\em randomized} construction algorithm was given in \cite{Nanongkai-STOC14} 
such that it was sufficient to solve \sssp on the resulting overlay network in order to solve the same problem on the whole network. We give a {\em deterministic} version of this result in \Cref{sec:overlay deterministic}. 
The second part is a more efficient algorithm for computing \sssp on an overlay network using
Procedures \ref{alg:priorities}, \ref{alg:clusters}, \ref{alg:hop_reduction_additive}, \ref{alg:hop_reduction}, and~\ref{alg:hop_set} from before (see \Cref{sec:distributed hop set}). In \Cref{sec:distributed final}, we show how to finish the computation after combining the two parts following \cite{Nanongkai-STOC14}. 
\Cref{fig:overlay2} illustrates our approach.

\subsection{Computing an Overlay Network Deterministically}\label{sec:overlay deterministic}

An {\em overlay network} (also known as {\em landmark} or {\em skeleton}  \cite{Sommer14,LenzenP_stoc13})
 as defined in \cite{Nanongkai-STOC14} 
 is a {\em virtual} network $G' = (V', E')$ of nodes $ V' $ and ``virtual edges'' $ E' $ that is built on top of an underlying {\em real} network $G = (V, E) $; i.e., $V' \subseteq V$ and
$E' = V' \times V'$ such that the weight of an edge in $G'$ is an approximation of the distance of its endpoints in $G$ and is $\infty$ if no path exists between
them in $G$. The nodes in $V'$ are called {\em centers}.
Computing~$ G' $ means that after the computation every node in $ G' $ knows whether it is a center and knows all virtual edges to its neighbors in $ G' $ and the corresponding weights.
We show in this subsection that there is a $\tilde O(\sqrt{n} / \epsilon + D)$-time algorithm that constructs an overlay network $G'$ of $\tilde O(\sqrt{n} / \epsilon)$ nodes such that a $(1 + \epsilon/3)$-approximation to \sssp in $G'$ can be converted to a $(1 + \epsilon)$-approximation to \sssp in $G$, as stated formally below.

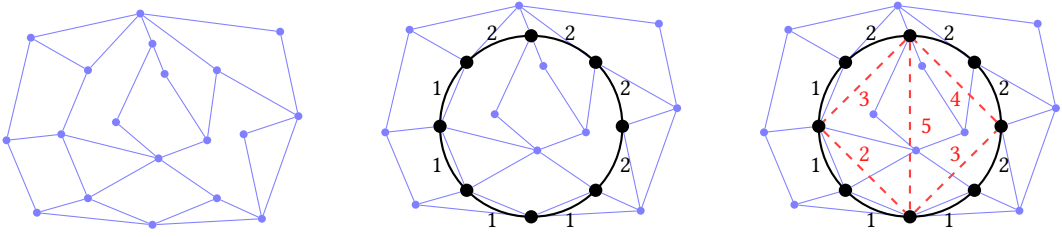
\begin{figure}
	\centering
	\begin{subfigure}[a]{0.32\textwidth}
	\begin{tikzpicture}[scale=0.8]
		\tikzstyle{node}=[circle,fill=blue!50,minimum size=3pt,inner sep=0pt,outer sep=0pt]
		\tikzstyle{edge} = [draw,color=blue!50]
		\tikzstyle{overlay-node}=[circle,fill=black,minimum size=5pt,inner sep=0pt,outer sep=0pt]
		\tikzstyle{hop-edge} = [draw,dashed,thick,color=red!75]

		% inner nodes
		\node[node] (u0) at (0.1,-0.4) {};
		\node[node] (u1) at (0.9,-0.1) {};
		\node[node] (u2) at (-0.6,0.2) {};
		\node[node] (u3) at (0.2,1.0) {};

		% outer nodes
		\node[node] (u4) at (2.4,0.3) {};
		\node[node] (u5) at (2.1,1.7) {};
		\node[node] (u6) at (-0.2,2.0) {};
		\node[node] (u7) at (-2.0,1.6) {};
		\node[node] (u8) at (-2.4,-0.1) {};
		\node[node] (u9) at (-1.9,-1.3) {};
		\node[node] (u10) at (1.8,-1.4) {};

		\draw[edge] (u0) -- (u1);
		\draw[edge] (u0) -- (u2);
		\draw[edge] (u1) -- (u3);
		\draw[edge] (u4) -- (u5);
		\draw[edge] (u4) -- (u10);
		\draw[edge] (u5) -- (u6);
		\draw[edge] (u6) -- (u7);
		\draw[edge] (u7) -- (u8);
		\draw[edge] (u8) -- (u9);

		\draw (0*360/8: 1.5cm) node[node] (v0) {};
		\draw (1*360/8: 1.5cm) node[node] (v1) {};
		\draw (2*360/8: 1.5cm) node[node] (v2) {};
		\draw (3*360/8: 1.5cm) node[node] (v3) {};
		\draw (4*360/8: 1.5cm) node[node] (v4) {};
		\draw (5*360/8: 1.5cm) node[node] (v5) {};
		\draw (6*360/8: 1.5cm) node[node] (v6) {};
		\draw (7*360/8: 1.5cm) node[node] (v7) {};

		\draw[edge] (v3) -- (v4);
		\draw[edge] (v4) -- (v5);
		\draw[edge] (v5) -- (v6);
		\draw[edge] (v6) -- (v7);

		\draw[edge] (v0) -- (u4);
		\draw[edge] (v0) -- (u10);
		\draw[edge] (v1) -- (u1);
		\draw[edge] (v1) -- (u4);
		\draw[edge] (v1) -- (u6);
		\draw[edge] (v2) -- (u2);
		\draw[edge] (v2) -- (u3);
		\draw[edge] (v2) -- (u6);
		\draw[edge] (v3) -- (u6);
		\draw[edge] (v3) -- (u7);
		\draw[edge] (v4) -- (u0);
		\draw[edge] (v4) -- (u8);
		\draw[edge] (v5) -- (u0);
		\draw[edge] (v5) -- (u9);
		\draw[edge] (v6) -- (u9);
		\draw[edge] (v6) -- (u10);
		\draw[edge] (v7) -- (u0);
		\draw[edge] (v7) -- (u10);
	\end{tikzpicture}
	\end{subfigure}
	\begin{subfigure}[a]{0.32\textwidth}
	\begin{tikzpicture}[scale=0.8]
		\tikzstyle{node}=[circle,fill=blue!50,minimum size=3pt,inner sep=0pt,outer sep=0pt]
		\tikzstyle{edge} = [draw,color=blue!50]
		\tikzstyle{overlay-node}=[circle,fill=black,minimum size=5pt,inner sep=0pt,outer sep=0pt]
		\tikzstyle{hop-edge} = [draw,dashed,thick,color=red!75]

		% inner nodes
		\node[node] (u0) at (0.1,-0.4) {};
		\node[node] (u1) at (0.9,-0.1) {};
		\node[node] (u2) at (-0.6,0.2) {};
		\node[node] (u3) at (0.2,1.0) {};

		% outer nodes
		\node[node] (u4) at (2.4,0.3) {};
		\node[node] (u5) at (2.1,1.7) {};
		\node[node] (u6) at (-0.2,2.0) {};
		\node[node] (u7) at (-2.0,1.6) {};
		\node[node] (u8) at (-2.4,-0.1) {};
		\node[node] (u9) at (-1.9,-1.3) {};
		\node[node] (u10) at (1.8,-1.4) {};

		\draw[edge] (u0) -- (u1);
		\draw[edge] (u0) -- (u2);
		\draw[edge] (u1) -- (u3);
		\draw[edge] (u4) -- (u5);
		\draw[edge] (u4) -- (u10);
		\draw[edge] (u5) -- (u6);
		\draw[edge] (u6) -- (u7);
		\draw[edge] (u7) -- (u8);
		\draw[edge] (u8) -- (u9);

		\draw (0*360/8: 1.5cm) node[overlay-node] (v0) {};
		\draw (1*360/8: 1.5cm) node[overlay-node] (v1) {};
		\draw (2*360/8: 1.5cm) node[overlay-node] (v2) {};
		\draw (3*360/8: 1.5cm) node[overlay-node] (v3) {};
		\draw (4*360/8: 1.5cm) node[overlay-node] (v4) {};
		\draw (5*360/8: 1.5cm) node[overlay-node] (v5) {};
		\draw (6*360/8: 1.5cm) node[overlay-node] (v6) {};
		\draw (7*360/8: 1.5cm) node[overlay-node] (v7) {};

		\draw[edge] (v3) -- (v4);
		\draw[edge] (v4) -- (v5);
		\draw[edge] (v5) -- (v6);
		\draw[edge] (v6) -- (v7);

		\draw[edge] (v0) -- (u4);
		\draw[edge] (v0) -- (u10);
		\draw[edge] (v1) -- (u1);
		\draw[edge] (v1) -- (u4);
		\draw[edge] (v1) -- (u6);
		\draw[edge] (v2) -- (u2);
		\draw[edge] (v2) -- (u3);
		\draw[edge] (v2) -- (u6);
		\draw[edge] (v3) -- (u6);
		\draw[edge] (v3) -- (u7);
		\draw[edge] (v4) -- (u0);
		\draw[edge] (v4) -- (u8);
		\draw[edge] (v5) -- (u0);
		\draw[edge] (v5) -- (u9);
		\draw[edge] (v6) -- (u9);
		\draw[edge] (v6) -- (u10);
		\draw[edge] (v7) -- (u0);
		\draw[edge] (v7) -- (u10);

		\draw[thick] (0, 0) circle (1.5cm);

		\draw (0.5*360/8: 1.675cm) node {\footnotesize $2$};
		\draw (1.5*360/8: 1.675cm) node {\footnotesize $2$};
		\draw (2.5*360/8: 1.675cm) node {\footnotesize $2$};
		\draw (3.5*360/8: 1.675cm) node {\footnotesize $1$};
		\draw (4.5*360/8: 1.675cm) node {\footnotesize $1$};
		\draw (5.5*360/8: 1.675cm) node {\footnotesize $1$};
		\draw (6.5*360/8: 1.675cm) node {\footnotesize $1$};
		\draw (7.5*360/8: 1.675cm) node {\footnotesize $2$};
	\end{tikzpicture}
	\end{subfigure}
	\begin{subfigure}[a]{0.32\textwidth}
	\begin{tikzpicture}[scale=0.8]
		\tikzstyle{node}=[circle,fill=blue!50,minimum size=3pt,inner sep=0pt,outer sep=0pt]
		\tikzstyle{edge} = [draw,color=blue!50]
		\tikzstyle{overlay-node}=[circle,fill=black,minimum size=5pt,inner sep=0pt,outer sep=0pt]
		\tikzstyle{hop-edge} = [draw,dashed,thick,color=red!75]

		% inner nodes
		\node[node] (u0) at (0.1,-0.4) {};
		\node[node] (u1) at (0.9,-0.1) {};
		\node[node] (u2) at (-0.6,0.2) {};
		\node[node] (u3) at (0.2,1.0) {};

		% outer nodes
		\node[node] (u4) at (2.4,0.3) {};
		\node[node] (u5) at (2.1,1.7) {};
		\node[node] (u6) at (-0.2,2.0) {};
		\node[node] (u7) at (-2.0,1.6) {};
		\node[node] (u8) at (-2.4,-0.1) {};
		\node[node] (u9) at (-1.9,-1.3) {};
		\node[node] (u10) at (1.8,-1.4) {};

		\draw[edge] (u0) -- (u1);
		\draw[edge] (u0) -- (u2);
		\draw[edge] (u1) -- (u3);
		\draw[edge] (u4) -- (u5);
		\draw[edge] (u4) -- (u10);
		\draw[edge] (u5) -- (u6);
		\draw[edge] (u6) -- (u7);
		\draw[edge] (u7) -- (u8);
		\draw[edge] (u8) -- (u9);

		\draw (0*360/8: 1.5cm) node[overlay-node] (v0) {};
		\draw (1*360/8: 1.5cm) node[overlay-node] (v1) {};
		\draw (2*360/8: 1.5cm) node[overlay-node] (v2) {};
		\draw (3*360/8: 1.5cm) node[overlay-node] (v3) {};
		\draw (4*360/8: 1.5cm) node[overlay-node] (v4) {};
		\draw (5*360/8: 1.5cm) node[overlay-node] (v5) {};
		\draw (6*360/8: 1.5cm) node[overlay-node] (v6) {};
		\draw (7*360/8: 1.5cm) node[overlay-node] (v7) {};

		\draw[edge] (v3) -- (v4);
		\draw[edge] (v4) -- (v5);
		\draw[edge] (v5) -- (v6);
		\draw[edge] (v6) -- (v7);

		\draw[edge] (v0) -- (u4);
		\draw[edge] (v0) -- (u10);
		\draw[edge] (v1) -- (u1);
		\draw[edge] (v1) -- (u4);
		\draw[edge] (v1) -- (u6);
		\draw[edge] (v2) -- (u2);
		\draw[edge] (v2) -- (u3);
		\draw[edge] (v2) -- (u6);
		\draw[edge] (v3) -- (u6);
		\draw[edge] (v3) -- (u7);
		\draw[edge] (v4) -- (u0);
		\draw[edge] (v4) -- (u8);
		\draw[edge] (v5) -- (u0);
		\draw[edge] (v5) -- (u9);
		\draw[edge] (v6) -- (u9);
		\draw[edge] (v6) -- (u10);
		\draw[edge] (v7) -- (u0);
		\draw[edge] (v7) -- (u10);

		\draw[thick] (0, 0) circle (1.5cm);

		\draw (0.5*360/8: 1.675cm) node {\footnotesize $2$};
		\draw (1.5*360/8: 1.675cm) node {\footnotesize $2$};
		\draw (2.5*360/8: 1.675cm) node {\footnotesize $2$};
		\draw (3.5*360/8: 1.675cm) node {\footnotesize $1$};
		\draw (4.5*360/8: 1.675cm) node {\footnotesize $1$};
		\draw (5.5*360/8: 1.675cm) node {\footnotesize $1$};
		\draw (6.5*360/8: 1.675cm) node {\footnotesize $1$};
		\draw (7.5*360/8: 1.675cm) node {\footnotesize $2$};

		\draw[hop-edge] (v0) -- (v2) node[midway,below=0.01pt,color=red!90] {\footnotesize $4$};
		\draw[hop-edge] (v2) -- (v4) node[midway,below=0.01pt,color=red!90] {\footnotesize $3$};
		\draw[hop-edge] (v4) -- (v6) node[midway,above=0.01pt,color=red!90] {\footnotesize $2$};
		\draw[hop-edge] (v6) -- (v0) node[midway,above=0.01pt,color=red!90] {\footnotesize $3$};
		\draw[hop-edge] (v2) -- (v6) node[midway,right=0.01pt,color=red!90] {\footnotesize $5$};
	\end{tikzpicture}
	\end{subfigure}
	\caption{An overview of the main steps of our algorithm. The left picture depicts the input graph. Thick edges and nodes (in black) in the middle picture depict a possible overlay network. Dashed edges (in red) in the right picture depict a possible hop set of the overlay network.}
	\label{fig:overlay2}
\end{figure}

\begin{table}
\centering
\begin{tabular}{c l}
$ G $						&	Graph defining underlying network, input of algorithm \\
$ G' $						&	Graph defining overlay network, output of algorithm \\
$ \epsilon $				&	Input parameter governing approximation quality \\
$ V $						&	Node set of $ G $, $ |V| = n $ \\
$ V' $						&	Node set $ G' $ called centers, $ V' \subseteq V $, $ |V'| = O(\sqrt{n} \lambda \log{(nW)} / \epsilon)$ \\
$ W $						&	Maximum edge weight of $ G $ \\
$ \diam $					&	Diameter of $ G $ \\
$ \lambda $					&	Number of bits used to represent each ID in the network $ G $ \\
$ \tilde{\epsilon} $		&	Parameter set to $ \tilde{\epsilon} = \tfrac{\epsilon}{10} $ \\
$ \hat{\epsilon} $			&	Parameter set to $ \hat{\epsilon} = \tfrac{\tilde{\epsilon}}{18 \lambda} $ \\
$ h $						&	Parameter set to $ h = \lfloor \hat{\epsilon} \sqrt{n} \rfloor $ \\
$ h' $						&	Parameter set to $ h' = (1 + 2 / \hat{\epsilon}) h $ \\
$ h^* $						&	Parameter set to $ h^* = 9 \lambda \sqrt{n} $ \\
$ k $						&	Parameter set to $ k = 2 h^* + 2 \sqrt{n} $ \\
$ k' $						&	Parameter set to $ k' = (1 + 2 / \hat{\epsilon}) k $ \\
$ j $						&	Generic index variable $ 0 \leq j \leq \lfloor \log{(n W)} \rfloor $ \\
$ \rho_j $					&	Rounding factor set to $ \rho_j = \tfrac{\hat{\epsilon} 2^j}{h} $ \\
$ \hat{G}_j $				&	Graph with edge weights $ w(u, v, \hat{G}_j) = \lceil \tfrac{w(u, v, G)}{\rho_j} \rceil $ \\
$ \varphi_j $				&	Rounding factor set to $ \varphi_j = \tfrac{\tilde{\epsilon} 2^j}{k} $ \\
$ \tilde{G}_j $				&	Graph with edge weights $ w(u, v, \tilde{G}_j) = \lceil \tfrac{w(u, v, G)}{\varphi_j} \rceil $ \\
$ \ball(v, \hat{G}_j, h') $	&	Ball of radius $ h' $ around $ v  $ in $ \hat{G}_j $ \\
$ t (v) $					&	Type of node $ v $, smallest $ j $ such that $ | \ball(v, \hat{G}_j, h') | \geq h $ \\
$ U_j $						&	Set of nodes of type $ j $, $ U_j \subseteq V $ \\
$ T_j $						&	$(2h'+1, (2h'+1) \lambda)$-ruling set $T_j$ for $\hat{G}_j$ of base set $ U_j $ \\
$ \dist^{h^*} (u, v, G) $	&	$h^*$-hop distance between $ u $ and $ v $ in $ G $ \\
$ \dist^k (u, v, G) $		&	$k$-hop distance between $ u $ and $ v $ in $ G $ \\
$ \hat{\dist} (u, v) $		&	$ (1 + \tilde{\epsilon}) $-approximation of $ \dist^k (u, v, G) $ \\
$ \tilde{\dist} (s, v) $	&	$ (1 + \epsilon/3) $-approximation of $ \dist(s, v, G') $
\end{tabular}
\caption{Overview of notation used in \Cref{sec:overlay deterministic}} \label{tab:notation subsection}
\end{table}

\begin{theorem}\label{thm:skeleton}\label{thm:overlay network}
In the broadcast \congest model, given any weighted undirected network $G = (V, E)$ with polynomially bounded positive integer edge weights and source node~$s$ and a parameter $ 0 < \epsilon \leq 1 $, there is an $ O(\sqrt{n} \lambda (\lambda + \log{(n W)}) \log{(n W)} / \epsilon + \diam) $-time deterministic distributed algorithm that computes an overlay network $G' = (V', E')$ and some additional information for every node with the following properties.  

\begin{itemize}
\item {\em Property 1:} $|V'|= O(\sqrt{n} \lambda \log{(nW)} / \epsilon)$ and $s\in V'$. 

\item {\em Property 2:} For every node $u\in V$, as soon as $ u $ receives a $(1 + \epsilon/3)$-approximation $\tilde{\dist} (s, v)$ of $\dist(s, v, G')$ for all centers $v\in V'$, it can infer a $(1 + \epsilon)$-approximation of $\dist(s, u, G)$ without any additional communication.
\end{itemize}
\end{theorem}
Note that in this paper we assume that $ \lambda = O (\log{n}) $ and $ W = \poly (n) $, and hence the overlay network has size $ |V'| = \tilde O (\sqrt{n} / \epsilon) $ and the running time is $ \tilde O (\sqrt{n} / \epsilon + \diam) $.
Observe that the statement of the theorem makes our algorithm very modular by separating the tasks of (i) constructing the overlay network and (ii) computing a $(1 + \epsilon/3)$-approximation of $\dist(s, v, G')$ for all centers $v\in V'$.
In \Cref{sec:distributed hop set} we show how to perform the second task by implementing the hop set algorithm of \Cref{sec:hop_set}.
It could, however, be replaced by any other algorithm providing such a $(1 + \epsilon/3)$-approximation, as is done, for example, in the recent approximate SSSP algorithm by Becker et al.~\cite{BeckerKKL16}, which also benefits from our deterministic construction of the overlay network.

Before proving the above theorem, we first recall how similar guarantees were achieved with a randomized algorithm in \cite{Nanongkai-STOC14}\footnote{We note that \cite{Nanongkai-STOC14} proved this theorem for general parameters $\lambda$ and $\alpha$, but we will only need it for $\lambda=\alpha=\sqrt{n}$.} (see Theorem 4.2 of the arXiv version\footnote{\url{https://arxiv.org/pdf/1403.5171v2.pdf}} of~\cite{Nanongkai-STOC14} for details).
\begin{itemize}
	\item In the first step of \cite{Nanongkai-STOC14}, the algorithm selects each node to be a center with probability $\tilde \Theta(1/\sqrt{n})$ and also makes $s$ a center. By a standard ``hitting set'' argument (e.g., \cite{UllmanY91,DemetrescuFI05}), any shortest path containing $\sqrt{n}$ edges will contain a center with high probability. Also, the number of centers is $\tilde \Theta(\sqrt{n})$ with high probability.
	\item In the second step, the algorithm makes sure that every node $v$ knows $ (1 + O(\epsilon)) $-approximate $\tilde \Theta(\sqrt{n})$-hop distances between $v$ and all centers using a {\em lightweight bounded-hop single-source shortest paths algorithm} from all centers in parallel, combined with the {\em random delay} technique to avoid congestion. 
\end{itemize}
Let us now give an overview of our new approach.
We derandomize the first step as follows: In \Cref{sec:types} we assign to each node $u$ a {\em type}, denoted by $t(u)$. (To compute these types, we invoke the {\em source detection algorithm} of Lenzen and Peleg \cite{LenzenP_podc13}, as we will explain in \Cref{sec:types}.) The important property of node types is that every path $\pi$ containing $\sqrt{n}$ edges contains a special node~$u$ of a ``desired'' type, meaning that $t(u)$ is not too big compared to $w(\pi, G)$ (see \Cref{thm:hitting paths} for details). This is  comparable to the property obtained from the hitting set argument, which would be achieved if we made the special node
of \emph{every} path a center. However, this may create too many centers (we want the number of centers to be $\tilde O(\sqrt{n} / \epsilon)$). Instead we select some nodes to be centers using the {\em ruling set} algorithm, as described in \Cref{sec:ruling set}. After this, we get a small set of centers such that every node $u$ of type $t(u)$ is not far from one of the centers. Thus, while we cannot guarantee that the path $\pi$ {\em contains} a center, we can guarantee that it contains a  node that is {\em not far} from a center (see \Cref{thm:ruling set} for details). 

To derandomize the second step, we use the recent algorithm of Lenzen and Patt-Shamir \cite{LenzenP14a-distance} for the {\em partial distance estimation} problem 
together with Procedures \ref{alg:priorities}, \ref{alg:clusters}, \ref{alg:hop_reduction_additive}, \ref{alg:hop_reduction}, and~\ref{alg:hop_set}, as we will explain in \Cref{sec:compute bounded hop distance}.  

Since this part of the paper is particularly dense in notation, we summarize the notation used in this subsection in \Cref{tab:notation subsection}.

\subsubsection{Types of Nodes}\label{sec:types}

Our algorithm initially spends $ O (\diam) $ rounds to make $ n $ and $ \lambda $ (and, if necessary, $ \epsilon $) global knowledge.
Every node internally sets $ \tilde{\epsilon} = \epsilon / 10 $, $ \hat{\epsilon} = \tilde{\epsilon} / (18 \lambda) $, $ h = \lfloor \hat{\epsilon} \sqrt{n} \rfloor $, and $ h' = (1 + 2 / \hat{\epsilon}) h $. Note that $ h' \leq 3 \sqrt{n} $.

For any integer $ 0 \leq j \leq \lfloor \log{n W} \rfloor $, we let $\rho_j = \tfrac{\hat{\epsilon} 2^j}{h}$ and let $\hat{G}_j$ be the graph with the same nodes and edges as $ G $ and weight $w(u, v, \hat{G}_j)=\lceil\tfrac{w(u, v, G)}{\rho_j}\rceil$ for every edge $(u, v)$.
Note that we have chosen~$ h' $ such that $ \dist(u, v, \hat{G}_j) \leq h' $ for all pairs of nodes $ u $ and $ v $ such that $ 2^j \leq \dist (u, v, G) \leq 2^{j+1} $ by~\eqref{eq:apsp approx main two} of \Cref{thm:property of weight rounding}.
For any node $u$, let the {\em ball} of $u$ in $\hat{G}_j$ be $\ball(u, \hat{G}_j, h') = \{v\in V) \mid \dist(u, v, \hat{G}_j)\leq h'\}$.
Note that for any index~$j$ and nodes $u$ and $v$, $\dist(u, v, \hat{G}_{j+1})\leq \dist(u, v, \hat{G}_j)$; thus, $\ball(u, \hat{G}_j, h')\subseteq \ball(u, \hat{G}_{j+1}, h')$. 
Let the {\em type} $t(u)$ of $u$ be the smallest index $j$ such that $|\ball(u, \hat{G}_j, h')|\geq h$. We crucially exploit the following structural property.

\begin{lemma}\label{thm:hitting paths}
For every path $\pi$ of $ G $ consisting of $|\pi|=\sqrt{n}$ edges there is a node $u$ on $\pi$ such that $2^{t(u)} \leq 2 \hat{\epsilon} w(\pi, G)$.
\end{lemma}

\begin{proof}
Let $\ell= \lceil |\pi|/h \rceil \geq 1/\eps'$, and let $ x $ and $ y $ denote the endpoints of $ \pi $.
Partition~$\pi$ into the path~$\pi_x$ consisting of the $(\ell-1)h$ edges closest to $x$ and the path~$\pi_y$ consisting of the $|\pi| - (\ell-1)h$ edges closest to $y$.
Further partition $\pi_x$ into $\ell-1$ nonoverlapping subpaths of exactly $h$ edges, and expand the path $\pi_y$ by adding edges of $\pi_x$ to it until it has $h$ edges.
Thus, there are now $\ell$ paths of exactly $h$ edges each and total weight at most $2 w(\pi, G)$.
It follows that there exists a subpath $\pi'$ of $ \pi $ consisting of exactly $h$ edges and weight at most $2 w(\pi, G)/\ell \leq 2 \hat{\epsilon} w(\pi, G)$.
Let $u$ and~$v$ be the two endpoints of $\pi'$, and let $ j $ be the index such that $2^j \leq \dist^h(u,v,G) \leq 2^{j+1}$.
By~\eqref{eq:apsp approx main two} of \Cref{thm:property of weight rounding} it follows that $\dist(u,v,\hat{G}_j) \leq h'$, which implies that $\ball(u, \hat{G}_j, h')$ contains $\pi'$.
Hence $|\ball(u, \hat{G}_j, h')|\geq h$ and $t(u) \leq j$. This shows that $2^{t(u)} \leq 2^j \leq \dist^h(u, v, G) \leq w(\pi', G) \leq 2\hat{\epsilon} w(\pi, G).$
\end{proof}

\paragraph{Computing Types of Nodes} 
To compute $t(u)$ for all nodes $u$, it is sufficient for every node~$u$ to know, for each $j$, whether $|\ball(u, \hat{G}_{j}, h')|\geq h$. We do this by solving the $(S, \gamma, \sigma)$-detection problem on $\hat{G}_j$ with $S=V$, $\gamma=h'$, and $\sigma=h$; i.e., we compute the list $ \CC (u, S, \gamma, \sigma, G) $ for all nodes $ u $, which contains the $ \sigma $ nodes from $S$ that are closest to $u$, provided their distance is at most $ \gamma $.
By \Cref{lem:source detection algorithm} this requires $O(\gamma + \sigma) = O(h+h') = O(\sqrt{n})$ rounds.
For any node $u$, $|\CC (u, V, h', h, G)|=h$ if and only if $|\ball(u, \hat{G}_j, h')|\geq h$. Thus, after we solve the $(S, \gamma, \sigma)$-detection problem on all~$\hat{G}_j$, using $ O(\sqrt{n} \log{(nW)})$ rounds, every node $u$ can compute its type $t(u)$ without any additional communication.

\subsubsection{Selecting Centers via Ruling Sets}\label{sec:ruling set} 

Having computed the types of the nodes, we compute ruling sets for the nodes of each type to select a small subset of nodes of each type as centers.
Remember the two properties of an $ (\alpha, \beta) $-ruling set~$ T $ of a base set~$ U $: (1) all nodes of $ T $ are at least distance $ \alpha $ apart, and (2) each node in $ U \setminus T $ has at least one ``ruling'' node of $ T $ in distance $ \beta $.
We use the algorithm of \Cref{lem:ruling set algorithm} to compute, for every $ 0 \leq j \leq \lfloor \log{n W} \rfloor $, a $(2h'+1, (2h'+1) \lambda)$-ruling set $T_j$ for $\hat{G}_j$ where the input set $U_j$ consists of all nodes of type $j$.
The number of rounds for this computation is $ O(h' \log{(nW)}) = O(\sqrt{n} \log{(nW)})$.
We define the set of centers as $V'=(\bigcup_{0 \leq j \leq \lfloor \log{(nW)} \rfloor} T_j) \cup \{s\}$.
Property (1) allows us to bound the number of centers, and by property (2) the centers ``almost'' hit all paths with $\sqrt{n}$ edges.

\begin{lemma}\label{thm:ruling set}
(1) The number of centers is $|V'|= O(\sqrt{n} \lambda \log{(nW)} / \epsilon)$. (2) For any path $\pi$ containing exactly $\sqrt{n}$ edges, there are a node $u$ in $\pi$ and a center $v\in V'$ such that $\dist^{h^*}(u, v, G) \leq \tilde{\epsilon} w(\pi, G))$, where $h^*=9 \sqrt{n} \lambda$.
\end{lemma}

\begin{proof}
(1) For each $j$, consider any two nodes $u$ and $v$ in $T_j$. Since $\dist(u, v, \hat{G}_j)> 2h'$ by Property (1) of the ruling set, $\ball(u, \hat{G}_j, h')\cap \ball(v, \hat{G}_j, h')=\emptyset$.
As every node $ u \in T_j $ is of type $ j $, $|\ball(u, \hat{G}_j, h')|\geq h$ for every $u\in T_j$.
We can therefore uniquely assign $ h $~nodes to every node $u\in T_j$, and thus $|T_j|\leq n/h = O(\sqrt{n} \lambda /\epsilon)$. 

(2) By \Cref{thm:hitting paths}, there is a node $u$ in $\pi$ such that $2^{t(u)} \leq 2 \hat{\epsilon} w(\pi, G) $. Moreover, there is a center~$v$ in the ruling set $T_{t(u)}$ such that
\begin{equation}
\dist(u, v, \hat{G}_{t(u)}) \leq (2h'+1) \lambda \leq 3 h' \lambda \leq h^* \, , \label{eq:type distance upper bound}
\end{equation}
where the second inequality is because $h'\leq 3\sqrt{n}$. 
Let $\pi'$ be the shortest path between $ u $ and $ v $ in $\hat{G}_{t(u)}$. Then $w(\pi', \hat{G}_{t(u)})= \dist(u, v, \hat{G}_{t(u)}) \leq h^*$, and as a consequence $\pi'$ contains at most $h^*$ edges. It follows that
\begin{align*}
\dist^{h^*}(u, v, G) &\leq w(\pi', G) && \text{(since $ \pi' $ is $u$-$v$ path with $ \leq h^* $ edges)} \\
&= \sum_{(x, y)\in E(\pi')}  w(x, y, G) \\
&\leq \sum_{(x, y)\in E(\pi')} \rho_{t(u)} \cdot w(x, y, \hat{G}_{t(u)}) && \text{(since $w(x, y, \hat{G}_{t(u)})=\lceil \tfrac{w(x, y, G)}{\rho_{t(u)}}\rceil$)} \\
&= \rho_{t(u)} \cdot w(\pi', \hat{G}_{t(u)}) \\
&= \rho_{t(u)} \cdot \dist(u, v, \hat{G}_{t(u)}) && \text{(since $ \pi' $ shortest $u$-$v$ path in $ \hat{G}_{t(u)} $)} \\
&= \frac{\hat{\epsilon} 2^{t(u)}}{h} \cdot \dist(u, v, \hat{G}_{t(u)}) && \text{(since $ \rho_{t(u)} = \tfrac{\hat{\epsilon} 2^{t (u)}}{h} $)} \\
&\leq \frac{\hat{\epsilon} 2^{t(u)}}{h} \cdot 3 h' \lambda && \text{(by \eqref{eq:type distance upper bound})} \\
&\leq 9 \lambda 2^{t(u)} && \text{(since $ h' = (1 + \tfrac{2}{\hat{\epsilon}}) h \leq \tfrac{3 h}{\hat{\epsilon}} $)} \\
&\leq 18 \lambda \hat{\epsilon} w(\pi, G) && \text{(by \Cref{thm:hitting paths})} \\
&= \tilde{\epsilon} w(\pi, G) && \text{(since $ \hat{\epsilon} = \tfrac{\tilde{\epsilon}}{18 \lambda} $).}
\end{align*}
\end{proof}

\subsubsection{Computing Distances to Centers} \label{sec:compute bounded hop distance}

Let $k=2h^*+2\sqrt{n}$, where $h^* = 9 \sqrt{n} \lambda $ (as in \Cref{thm:ruling set}), and let $ k' = (1 + 2 / \hat{\epsilon}) k $.
In this step, we compute for every node $ u $ and every center $ v $ a value $ \distest (u, v) $ that is a $(1+\tilde{\epsilon})$-approximation of $\dist^{k}(u, v, G)$ such that each node $ u $ knows $ \distest (u, v) $ for all centers~$ v $.
In particular, we also compute $ \distest (u, v) $ for all pairs of centers $ u $ and $ v $. To do this we follow the idea of partial distance estimation~\cite{LenzenP14a-distance}.  As in \Cref{sec:types}, we do this by solving the source detection problem on a graph with rounded weights.\footnote{We note that the algorithm and analysis described in this subsection are essentially the same as in the proof of~\cite[Theorem 3.3]{LenzenP14a-distance}. We cannot use the result in \cite{LenzenP14a-distance} directly since we need a slightly stronger guarantee, which can already be achieved by the same proof. (We thank Christoph Lenzen for a communication regarding this.)} For every integer $0 \leq j \leq \lfloor \log{n W} \rfloor$, let $\varphi_j = \tfrac{\tilde{\epsilon} 2^j}{k}$ and let $\tilde{G}_j$ be the weighted graph such that $w(u, v, \tilde{G}_j)=\lceil\tfrac{w(u, v, G)}{\varphi_j}\rceil$ for every edge $(u, v)$ in $G$.

We solve the $(S, \gamma, \sigma)$-detection problem on $\tilde{G}_j$ for all $0 \leq j \leq \lfloor \log{nW} \rfloor$, with parameters $S=V'$, $\gamma=k'= (1 + 2 / \hat{\epsilon}) k = O(\sqrt{n} \lambda^2) $, and $\sigma=|V'|$, where $|V'|= O(\sqrt{n} \lambda \log{(nW)} / \epsilon)$ by \Cref{thm:ruling set}.
Using the algorithm of \Cref{lem:source detection algorithm} for each graph~$\tilde{G}_j$ this takes $ O((\gamma + \sigma) \log{(nW)})= O(\sqrt{n} \lambda (\lambda + \log{(n W)}) \log{(n W)} / \epsilon)$ rounds. At termination, every node $u$ knows the distances up to distance range $k'$ to all centers in all~$\tilde{G}_j$; i.e., it knows $\dist(u, v, k', \tilde{G}_j)$ for all $j$ and all centers $v$. For every node $u \in V$ and every center $v \in V'$ we set $\distest(u, v)=\min_{0 \leq j \leq \lfloor \log{n W} \rfloor} \{\varphi_j \cdot \dist(u, v, k', \tilde{G}_j)\}$. Every node $u$ can compute $\distest(u, v)$ without any additional communication as soon as the source detection algorithm is finished.
Now consider the index $j^*$ such that $ 2^{j^*} \leq \dist^h(u, v, G)\leq 2^{j^*+1} $.
It follows from~\eqref{eq:apsp approx main two} of~\Cref{thm:property of weight rounding} that $\dist(u,v,\tilde{G}_{j^*}) \leq k'$ which implies that $\dist(u,v,k',\tilde{G}_{j^*}) = \dist(u,v,\tilde{G}_{j^*})$.
With \eqref{eq:apsp approx main one} and~\eqref{eq:apsp approx main three} we then get
\begin{equation}
\distest(u, v) \leq \varphi_{j^*} \cdot \dist(u, v, k', \tilde{G}_{j^*}) = \varphi_{j^*} \cdot \dist(u, v, \tilde{G}_{j^*}) \leq (1+\tilde{\epsilon})\dist^k(u, v, G) \, . \label{eq:dhat upper bound}
\end{equation}
Hence $\distest(u, v)$ is the desired $(1+\tilde{\epsilon})$-approximation of $\dist^{k}(u, v, G)$.

\subsubsection{Completing the Proof of \Cref{thm:overlay network}}

We define our final overlay network to be the graph $G'$ where the weight between any two centers $u, v\in V'$ is $\distest(u, v)$ (as computed in \Cref{sec:compute bounded hop distance}). Additionally, for every node $ u \in V $ we store the value of $\distest(u, v)$ for all centers $v\in V'$.
We now show that all properties stated in \Cref{thm:overlay network} hold for $G'$. Since we need $ O(\sqrt{n} \log{(nW)})$ rounds in \Cref{{sec:types,sec:ruling set}} and $ O(\sqrt{n} \lambda (\lambda + \log{(n W)}) \log{(n W)} / \epsilon)$ rounds in \Cref{sec:compute bounded hop distance}, the running time to construct $G'$ is $O(\sqrt{n} \lambda (\lambda + \log{(n W)}) \log{(n W)} / \epsilon)$. Moreover, $|V'|=O(\sqrt{n} \lambda \log{(nW)} / \epsilon)$ as shown in \Cref{thm:ruling set}. This is as claimed in the first part of \Cref{thm:overlay network}.
It is thus left to prove the following statement in \Cref{thm:overlay network}: 
``for every node $u\in V$, as soon as $ u $ receives a $(1 + \epsilon/3)$-approximation $\tilde{\dist} (s, v)$ of $\dist(s, v, G')$ for all centers $v\in V'$, it can infer a $(1 + \epsilon)$-approximate value of $\dist(s, u, G)$ without any additional communication.''
Recall that in \Cref{sec:distributed hop set} we show how to compute, and make known to all nodes, the values $ \tilde{\dist} (s, v) $ for all centers $v\in V'$.

Consider any node $u$, and let $\pi$ be the shortest path between $s$ and $u$ in $G$. If $\pi$ contains fewer than $\sqrt{n}$ edges, then $ \dist^k (s, u, G) = \dist (s, u, G) $ and thus the value $\distest(s, u)$, which is a $(1+\tilde{\epsilon})$-approximation of $\dist^k(s, u, G)$ known by $u$, is already a $(1+\tilde{\epsilon})$-approximation of $\dist(s, u, G)$ (and thus, by the choice of $ \tilde{\epsilon} = \epsilon / 10 $, also a $ (1 + \epsilon) $-approximation).
If $\pi$ contains at least $\sqrt{n}$ edges, then partition $\pi$ into subpaths $\pi_0, \pi_1, \ldots, \pi_\ell$ (for some $\ell \geq 0$), where $\pi_0$ contains~$s$, $\pi_\ell$ contains~$u$, $\pi_0$ contains at most $\sqrt{n}$ edges, and every subpath except $\pi_0$ contains exactly $\sqrt{n}$ edges. By \Cref{thm:ruling set}, for every $1\leq i\leq \ell$, there are a node $x_i$ and a center~$y_i$ such that (i) $x_i$ is in $\pi_i$, and (ii) $\dist^{h^*}(x_i, y_i, G) \leq \tilde{\epsilon} w(\pi_i, G) $. Additionally, for $ 1 \leq i \leq \ell - 1 $, since $x_i$ and $x_{i+1}$ lie on $\pi$, their shortest path is the subpath of $\pi$ between them and, thus, it consists of at most $2 \sqrt{n}$ edges.
It follows that $\dist^{2\sqrt{n}}(x_i, x_{i+1}, G) = \dist(x_i, x_{i+1}, G)$.
By our choice of $ k = 2h^* + 2 \sqrt{n} $, the triangle inequality and symmetry (i.e., $ \dist (v, v', G) = \dist (v', v, G) $) give, for every $ 1 \leq i \leq \ell - 1 $,
\begin{equation}\label{eq:bound on close-by nodes}
\begin{split}
\dist^k(y_i, y_{i+1}, G)&\leq \dist^{h^*}(y_i, x_i, G) + \dist^{2\sqrt{n}}(x_i, x_{i+1}, G)+\dist^{h^*}(x_{i+1}, y_{i+1}, G) \\
 &\leq \tilde{\epsilon} w(\pi_i, G) + \dist(x_i, x_{i+1}, G) + \tilde{\epsilon} w(\pi_{i+1}, G) \, .
\end{split}
\end{equation}
By the same argument,
\begin{equation}\label{eq:bound on close-by nodes start}
\dist^k(s, y_1, G) \leq \dist^{2 \sqrt{n}} (s, x_1, G) + \dist^{h^*} (x_1, y_1, G) \leq \dist(s, x_1, G) + \tilde{\epsilon} w (\pi_1, G)
\end{equation}
and
\begin{equation}\label{eq:bound on close-by nodes end}
\dist^k(y_\ell, u, G) \leq \dist^{h^*} (y_\ell, x_\ell, G) + \dist^{\sqrt{n}} (x_\ell, u, G) \leq \dist(x_\ell, u, G) + \tilde{\epsilon} w(\pi_\ell, G) \, .
\end{equation}

We now argue that $ \tilde{\dist} (s, y_\ell) + \distest (u, y_\ell) $ (the sum of two values known to node $ u $) is a $ (1 + \epsilon) $-approximation of $ \dist(s, u, G) $.
First, since $ \tilde{\dist} (s, y_\ell) $ is a $ (1 + \epsilon/3) $-approximation of $ \dist (s, y_\ell, G') $, we get
\begin{equation*}
\tilde{\dist} (s, y_\ell) + \distest (u, y_\ell) \leq (1 + \epsilon/3) \dist (s, y_\ell, G') + \distest (u, y_\ell) \leq (1 + \epsilon/3) (\dist (s, y_\ell, G') + \distest (u, y_\ell)) \, .
\end{equation*}
We now apply the triangle inequality, exploit that every edge $ (x, y) $ in $ G' $ has weight $ \distest (x, y) $ (implying $ \dist (x, y, G') \leq \distest (x, y) $), and use \eqref{eq:dhat upper bound}--\eqref{eq:bound on close-by nodes end} from above to get % {eq:dhat upper bound,eq:bound on close-by nodes,eq:bound on close-by nodes start,eq:bound on close-by nodes end}
\begin{align*}
\dist (s, y_\ell, G') + \distest (u, y_\ell) &\leq \left( \dist (s, y_1, G') + \sum_{i=1}^{\ell-1} \dist (y_i, y_{i+1}, G') \right) + \distest (u, y_\ell) \\
 &\leq \left( \distest (s, y_1) + \sum_{i=1}^{\ell-1} \distest (y_i, y_{i+1}) \right) + \distest (u, y_\ell) \\
 &\stackrel{\mathclap{\eqref{eq:dhat upper bound}}}{\leq} (1 + \tilde{\epsilon}) \left( \dist^k (s, y_1, G) + \sum_{i=1}^{\ell-1} \dist^k (y_i, y_{i+1}, G) \right) + (1 + \tilde{\epsilon}) \dist^k (u, y_\ell) \\
 &\leq (1 + \tilde{\epsilon}) \left( \dist^k (s, y_1, G) + \sum_{i=1}^{\ell-1} \dist^k (y_i, y_{i+1}, G) + \dist^k (u, y_\ell, G) \right) \\
 &\stackrel{\mathclap{\text{\eqref{eq:bound on close-by nodes}--\eqref{eq:bound on close-by nodes end}}}}{\leq} (1 + \tilde{\epsilon}) \left( \dist(s, u, G) + 2 \tilde{\epsilon} \sum_{i=1}^\ell w(\pi_i, G) \right) \\
 &\leq (1 + \tilde{\epsilon}) \left( \dist(s, u, G) + 2 \tilde{\epsilon} \dist(s, u, G) \right)\\
 &= (1 + \tilde{\epsilon}) (1 + 2 \tilde{\epsilon}) \dist(s, u, G) \\
 &\leq (1 + 5 \tilde{\epsilon}) \dist(s, u, G) \\
 &= (1 + \epsilon/2) \dist(s, u, G) \, .
\end{align*}
By combining the two derivations above, we get
\begin{equation*}
\tilde{\dist} (s, y_\ell) + \distest (u, y_\ell) \leq (1 + \epsilon/3) (1 + \epsilon/2) \dist(s, u, G) \leq (1 + \epsilon) \dist(s, u, G) \, .
\end{equation*}
Thus, when $u$ receives $ \tilde{\dist} (s, v') $ for all centers $ v \in V' $,
it can compute the value $\min_{v \in V'} (\tilde{\dist} (s, v) + \distest(u, v)$) and, as $ y_\ell \in V' $, the argument above shows that this value is a $(1 + \epsilon)$-approximation of $\dist (s, u, G)$.

\subsection{Computing a Hop Set on an Overlay Network}\label{sec:distributed hop set}

We now show how to implement the algorithm to compute the hop set on the overlay network~$G'$ presented in \Cref{sec:hop_set} and how to compute approximate shortest paths from $ s $ in $ G' $ using the hop set presented in \Cref{sec:hop_set}.
We let $ G' $ be the overlay network obtained from \Cref{thm:overlay network} with $\epsilon=1/\log n$ (to guarantee a $ (1 + o(1)) $-approximation in the end).
Throughout the algorithm we will work on overlay networks whose node set is the set of centers $V'$, but which might have different edge weights as, e.g., Procedure~\ref{alg:hop_reduction} calls Procedure~\ref{alg:hop_reduction_additive} on overlay networks with modified edge weights.
Thus, we will use $G''$ to refer to an overlay network (with set of nodes $ V' $) on which
Procedures~\ref{alg:priorities}, \ref{alg:clusters}, \ref{alg:hop_reduction_additive}, and \ref{alg:hop_reduction} run to emphasize the fact that they might not equal $G'$. 
We let $N$ be the number of centers in $G'$ and~$G''$.
Thus $N=\tilde O(\sqrt{n})$.

\subsubsection{Computing Bounded-Distance Single-Source Shortest Paths}\label{sec:distributed bounded depth}

We will repeatedly use an algorithm for computing a shortest-path tree up to distance~$R$ rooted at some center $s$ on an overlay network $ G'' $, where $R = N^{o(1)}$.
At the end of the algorithm every center knows this tree. We do this in a breadth-first search manner, in $R+1$ {\em iterations}.
As in Dijkstra's algorithm, every center keeps a tentative distance $ \delta (s, u) $ from $ s $ and a tentative parent in the shortest-path tree, i.e., a center $ v $ such that $ \delta (s, u) = \delta (s, v) + w(u, v, G'') $.
Initially, $ \delta (s, s) = 0 $ and $ \delta (s, v) = \infty $ for every center $ v \neq s $.
In the $L$th iteration, for $ L $ from $ 0 $ up to~$ R $, all centers in $G''$ whose tentative distance $\delta(s, u)$ is exactly $L$ broadcast\footnote{More precisely, there is a designated node (e.g., the node with lowest ID) that aggregates and distributes the messages (via upcasting and downcasting on the breadth-first search tree of the underlying network $ G $) and tells other centers when the iteration starts and ends.} to all other centers a message $(u, \delta(s, u), v)$, where $ v $ is the parent of~$ u $. Using this information, every center $u$ will update (``relax'') its tentative distance $\delta(s, u)$ and its tentative parent.

By a straightforward induction, after the $L$th iteration, centers that have distance~$L$ from~$s$ (i.e., that are at level $ L $ in the shortest-path tree) will already know their correct distance. 
Thus, at the end of the last iteration every center knows the shortest-path tree rooted at~$s$ up to distance~$R$ in~$G''$.
To analyze the running time, note that over $R$ rounds we broadcast $ N $ messages in total, and if $M_L$ messages are broadcast in the $L$th iteration, then this iteration takes $O(M_L+\diam)$ rounds. (We emphasize that the number of rounds depends on the diameter $ \diam $ of the original network, and not of~$G''$.)
The total number of communication rounds used over all iterations is thus $O(R\diam+\sum_L M_L) = O(N + R\diam)$.

\subsubsection{Computing Priorities} \label{sec:distributed priorities}
We implement Procedure~\ref{alg:priorities} on an overlay network~$ G'' $.
All necessary parameters can be computed beforehand and thus do not require any communication.
Initially every center knows that it is contained in $ A_0 = V' $.
To compute $ A_{i+1} $ given that $ A_i $ is known (i.e., every center knows whether or not it is in~$ A_i $), we compute the proximity list $ \CC (v, A_i, R, q, G'') $ for every center $ v $ using a source detection algorithm and distribute each list to every center, where, by our choice of parameters, $R = N^{o(1)}$ and $q = N^{o(1)}$.
Then every center runs the same deterministic greedy hitting set approximation algorithm to compute $ A_{i+1} $.\footnote{Note that the number of internal computation steps of the greedy algorithm at each center is linear in its input, which we can upper-bound by $ O (N q) = n^{1/2+o(1)} $.}
We will obtain $ \mathcal{A} = (A_i)_{0\leq i\leq p}$ by repeating this for $p$ iterations. 
Thus, we have to solve the $(S, \gamma, \sigma) $-source detection problem with $ S = A_i $, $ \gamma = R $, and $ \sigma = q $ on an overlay network $ G'' $.
For this purpose we simulate the source detection algorithm of Lenzen and Peleg~\cite{LenzenP_podc13} (see \Cref{lem:source detection algorithm} and the preceding description of the algorithm) as if run on the overlay network.

The simulated source detection algorithm consists of at most $ \gamma + \sigma $ iterations. Since in each iteration of the source detection algorithm, each center sends \emph{the same} message to all of its neighbors, we can simulate each iteration by broadcasting at most $ N $ messages in the underlying network $ G $.
Thus, simulating the source detection algorithm takes $O((\sigma + \gamma) (N + \diam))$ rounds.
To compute the priorities, we repeat this process for all $ p \leq \log{n} $ priorities.
The overall running time for implementing Procedure~\ref{alg:priorities} therefore is $O((N + \diam) p (R + q))$.
With our choice of parameters ($N=\tilde O(\sqrt{n})$, $ p \leq \log{n} $, $q=N^{o(1)}$, and $R=N^{o(1)}$), this becomes
\begin{equation}\label{eq:computing priorities overlay}
n^{1/2+o(1)} + \diam n^{o(1)} \, .
\end{equation}

\subsubsection{Computing Clusters} 
We now describe how to compute clusters on an overlay network~$ G'' $ such that at the end of this computation, every center will know all clusters $ \clust (v, \cA, R, G'') $ (i.e., its own cluster and the cluster of every other center).
We do this by implementing Procedure~\ref{alg:clusters} on $ G'' $. First, we need to compute $\dist (v, A_{i+1}, R, G'')$, for every $1\leq i\leq p$. This can be done in exactly the same way as in the first phase of computing the hierarchy $\mathcal{A} = (A_i)_{0 \leq i \leq p}$; i.e., we add a virtual source $s^*$ and edges of weight zero between $s^*$ and centers in $A_i$, and compute the shortest-path tree up to distance~$R$ rooted at $s^*$. Since such a tree can be computed in $ O (N + R\diam) $ rounds and we have to compute $ p \leq \log{n} $ such trees, the total time we need here is $ \tilde O (N + R\diam) $. 

Next, we use the information gained above to compute the cluster up to distance~$R$ from every center $u$ in~$G''$, as described in Procedure~\ref{alg:clusters}. That is, in iteration $L$ (starting with $L=0$ and ending with $L=R$),
every center $v$ having (i) $\delta(u, v)$ (the tentative distance from $u$ to $ v $) equal to $L$
and (ii) additionally $ \delta (u, v) < \dist (v, A_{i+1}, R, G'') $
 will broadcast\footnote{We note again that to do this, there is a designated center that aggregates and distributes the messages (via upcasting and downcasting), and tells other centers when the iteration starts and ends.} its distance to $ u $ to all other centers so that every other center, say $w$, can (i) update its tentative distance $\delta(u, w)$ and
(ii) add $v$ and $ \delta (u, v) $ to its locally stored copy of $C(u)$.
Thus, there are $\sum_{v \in V'} | \clust (v, \cA, R, G'') |$ messages broadcast in total, which is bounded from above by $\tilde O(pN^{1/p}) = n^{1/2+o(1)}$ due to \Cref{thm:clusters}. 

Note that this procedure computes $ \clust (v, \cA, R, G'') $, for all centers $v$, {\em in parallel}. Each iteration $L$ requires $O(\sum_{v\in V'} M_{v, L} + \diam)$ rounds, where $M_{v, L}$ is the number of messages broadcast by node $ v $ in iteration $L$ in the above computation. The total number of rounds over all $R$ iterations is thus 
\begin{equation*}
O \left( \sum_{0\leq L\leq R}\sum_{v\in V'} M_{v, L} + R\diam \right) = O\left( \sum_{v\in V'} |\clust (v, \cA, R, G'')| + R\diam \right) =  n^{1/2+o(1)} + \diam n^{o(1)} \,.\label{eq:cluster time}
\end{equation*}
Note that since the computation is done by broadcasting messages, every center knows the cluster $ \clust (v, \cA, R, G'') $ for all $v$ at the end of this computation.
Together with the running time bound of~\eqref{eq:computing priorities overlay} for computing the priorities, we arrive at the following guarantees.
\begin{lemma}\label{thm:distributed cluster}
For any overlay network $G'' = (V', E'')$ with $N = \tilde O(\sqrt{n})$ centers, the above algorithm, in $n^{1/2+o(1)} + \diam n^{o(1)}$ rounds, deterministically computes a hierarchy of centers $ \cA = (A_i)_{0 \leq i \leq p} $ and clusters $ \clust (v, \cA, R, G'') $ for each center $ v $ as specified in \Cref{thm:clusters} with $ p \leq \log{n} $ priorities up to distance $ R = N^{o(1)} $ such that $ \sum_{v \in V'} | \clust (v, \cA, R, G'') | = n^{1 + o(1)} $ (and every center knows $\clust (v, \cA, R, G'')$ for all centers $v$ as well as the value of $ \dist (v, w, G'') $ for every center $ v $ and every center $ w \in \clust (v, \cA, R, G'') $).
\end{lemma}

%Given a weighted graph $ G = (V, E) $ with positive integer edge weights and parameters $ p \geq 2 $ and $ R \geq 1 $, Procedure~\ref{alg:clusters} computes a hierarchy of sets $ \cA = (A_i)_{0 \leq i \leq p} $, where $ V = A_0 \subseteq A_1 \subseteq \dots \subseteq A_p = \emptyset $, such that $ \sum_{v \in V} | \clust (v, \cA, R, G) | = \tilde O( p n^{1 + 1/p}) $.
%It also computes for every node~$ v $ the set $ \clust (v, \cA, R, G) $ and for each node $ w \in \clust (v, \cA, R, G) $ the value of $ \dist (v, w, G) $.

\subsubsection{Computing the Hop Reduction with Additive Error}
We implement Procedure~\ref{alg:hop_reduction_additive} on an overlay network~$ G'' $.
All necessary parameters can be computed beforehand and thus require no communication.
We then execute \Clusters{$ G'' $, $ p $, $ R $} using the above algorithm to get $ (\clust (v, \cA, R, G''), \delta (v, \cdot))_{v \in V} $.
With this information, the set $ F $, as specified in Procedure~\ref{alg:hop_reduction_additive}, can be computed without any additional communication.
Thus, executing \Clusters{$ G'' $, $ p $, $ R $} is the only part of computing~$ F $ that requires communication.
By \Cref{thm:distributed cluster} the total time needed to execute Procedure~\ref{alg:hop_reduction_additive} is therefore
\begin{equation}
n^{1/2+o(1)} + \diam n^{o(1)} \, .\label{eq:distributed_hop_reduction_additive_error} 
\end{equation}

\subsubsection{Computing the Hop Reduction without Additive Error}
We implement Procedure~\ref{alg:hop_reduction} on an overlay network~$ G'' $.
All necessary parameters can be computed beforehand and thus do not require any communication.
Moreover, every center knows about the edges incident to it, and we can thus implicitly compute $ \hat{G}_j $, as specified in Procedure~\ref{alg:hop_reduction}, by scaling down edge weights without any communication.
We then execute Procedure~\ref{alg:hop_reduction_additive} to compute~$ \hat{F}_j $.
Knowing $ \hat{F}_j$, we can compute $ F $ without any additional communication.
Thus, executing Procedure~\ref{alg:hop_reduction_additive} is the only part of computing~$ F $ 
that requires communication, and it is executed $ O (\log{(nW)}) $ times.
As our implementation of Procedure~\ref{alg:hop_reduction_additive} takes time $ n^{1/2+o(1)} + \diam n^{o(1)} $, as argued above (cf.~\eqref{eq:distributed_hop_reduction_additive_error}), the total time needed to execute Procedure~\ref{alg:hop_reduction} is
\begin{equation}
n^{1/2+o(1)} \log{W} + \diam n^{o(1)} \log{W} \, . \label{eq:distributed_hop_reduction} 
\end{equation}

\subsubsection{Computing the Hop Set}
We implement Procedure~\ref{alg:hop_set} on the overlay network~$ G' $.
All necessary parameters can be computed beforehand.
Computing $ F_{i+1} $ is done by calling Procedure~\ref{alg:hop_reduction} on the graph $ H_i $, as specified in Procedure~\ref{alg:hop_set}, which, as argued above (cf.~\eqref{eq:distributed_hop_reduction}), takes time $ n^{1/2+o(1)} \log{W} + \diam n^{o(1)} \log{W} $.
As every center knows its incident edges, the graph $ H_{i+1} $ can be computed from $ F_{i+1} $ without any additional communication.
As we execute Procedure~\ref{alg:hop_reduction} $ p \leq \log{n} $ times, the total time needed to implement Procedure~\ref{alg:hop_set} is
$ p n^{1/2+o(1)} \log{W} + p \diam n^{o(1)} \log{W} = n^{1/2+o(1)} \log{W} + \diam n^{o(1)} \log{W} $.
By running this algorithm on $G'$ (which, as pointed out, involves performing hop reductions and computing clusters on some other overlay networks), we obtain the following theorem.

\begin{theorem}
In the broadcast \congest model, there is a deterministic algorithm that, for any overlay network $G'$ with $N = \tilde O(\sqrt{n})$ centers and positive integer weights in the range $ \{1, \ldots, W \} $ on edges between centers, computes an $ (n^{o(1)}, o(1)) $-hop set of~$ G' $ in $ n^{1/2 + o(1)} \log{W} + D n^{o(1)} \log{W} $ rounds.
When the algorithm has finished, every center knows every edge in the hop set.
\end{theorem}

\subsubsection{Routing via the Hop Set}

Remember that the overlay network is computed using the source detection algorithm of Lenzen and Peleg~\cite{LenzenP_podc13}.
If a node $ x $ of the overlay network wants to send a message to one of its neighbors~$ y $ in the overlay network, it can do so by routing the message along a path in the original network whose length is upper-bounded by the weight of the overlay edge $ (x, y) $.
This routing can be obtained by modifying the source detection algorithm to additionally construct breadth-first search trees rooted at the sources (see~\cite{LenzenP_podc13}), which in our case are the nodes of the overlay network.

When we compute the hop set on the overlay network, we broadcast all computed clusters to all nodes in the network.
In this way the clusters, the corresponding partial shortest-path trees of the clusters, as well as the hop set edges become global knowledge.
Therefore every node in the overlay network learns for every hop set edge $ (x, y) $ its corresponding path from $ x $ to $ y $ in the overlay network.
Thus, also for every hop set edge $ (x, y) $ of the overlay network, $ x $ can send a message to $ y $ by routing the message along a path in the original network whose length is upper-bounded by the weight of the overlay edge $ (x, y) $.
This means that the hop set computed by our algorithm has the following \emph{path-reporting} property, as introduced in~\cite{ElkinN-PODC16}: A hop set $ F $ for a graph $ G $ is called \emph{path-reporting} if for every hop set edge $ (x, y) \in F $ of weight $ b $ there exists a corresponding path $ \pi $ in $ G $ between $ x $ and $ y $ of length $ b $. Furthermore, every node $ v $ on $ \pi $ knows $ \dist_\pi (v, x) $ and $ \dist_\pi (v, y) $ and its neighbors on $ \pi $.

\subsection{Final Steps}\label{sec:distributed final}

Let $H = G' \cup F$ be the graph obtained by adding to $G'$ the edges of the $(n^{o(1)}, o(1))$-hop set $ F $ computed above. To $(1+o(1))$-approximate $\dist(s, v, G')$ for every center $v$ in $G'$, it is sufficient to $(1+o(1))$-approximate the $h$-hop distance $\dist^h(s, v, H)$ for some $h=n^{o(1)}$. The latter task can be done in $O(h\diam+|V'|)= n^{o(1)} \diam + n^{1/2+o(1)} $ rounds by the same method as in Lemma~4.6 in the arXiv version\footnote{\url {https://arxiv.org/pdf/1403.5171v2.pdf}} of~\cite{Nanongkai-STOC14}. 
We give a sketch here for completeness. Let $\epsilon=1/\log n$. For any $0 \leq j \leq \lfloor \log(nW) \rfloor$, let $\hat{H}_j$ be the graph obtained by rounding edge weights in $H$ as in \Cref{sec:types}; i.e., for every edge $ (u, v)$ we set $w(u, v, \hat{H}_j)=\lceil \tfrac{w(u, v, H)}{\rho_j}\rceil,$ where $\rho_j=\tfrac{\epsilon 2^j}{h}$.
For each $\hat{H}_j$, we compute the shortest-path tree rooted at $s$ up to distance $R=O(h/\epsilon)$, which can be done in $ R\diam + n^{1/2+o(1)} = n^{o(1)}\diam + n^{1/2+o(1)}$ rounds, using the algorithm described in \Cref{sec:distributed bounded depth}. This gives $\dist(s, v, R, \hat{H}_j)$ for every center~$v$. 
We then use the following value as $ (1 + o(1)) $-approximation of $\dist^h(s, v, H)$ (and thus of $\dist(s, v, G')$): $\tilde{d} (s, v) = \min_j \rho_j \cdot \dist(s, v, R, \hat{H}_j)$. 
The correctness of this algorithm follows from \Cref{thm:property of weight rounding}.

Once we have $(1+o(1))$-approximate values of $\dist(s, v, G')$ for every center $v\in V'$, we can broadcast these values to the whole network in $\tilde O(\sqrt{n} + \diam)$ rounds. \Cref{thm:overlay network} then implies that we have a $(1+o(1))$-approximate solution to the \sssp problem on the original network. The total time spent is $ n^{1/2+o(1)} + \diam n^{o(1)} $. By observing that the term $n^{o(1)} \diam $ will show up in the running time only when $\diam=\omega(n^{o(1)})$, we can write the running time as $ n^{1/2+o(1)} + \diam^{1+o(1)}$, as claimed in the beginning.

We thus have obtained the following result.
\begin{theorem}
In the broadcast \congest model, there is a deterministic algorithm that, on any weighted undirected network with polynomially bounded positive integer edge weights, computes $ (1 + o(1)) $-approximate shortest paths between a given source node $ s $ and every other node in $  n^{1/2+o(1)}+\diam^{1+o(1)} $ rounds.
\end{theorem}

\section{Algorithms in Other Settings}\label{sec:other results}

\subsection{Congested Clique}

In the congested clique model, the underlying communication network is a complete graph.
Thus, in each round every node can send a message to every other node.
Apart from this topological constraint, the congested clique model is similar to the \congest model.

We compute an $(n^{o(1)}, o(1))$-hop set on a congested clique by implementing the hop set construction algorithm in the same way as on the overlay network, as presented in \Cref{sec:distributed hop set}. (However, we do {\em not} compute an overlay network here.) The only difference is the number of rounds needed for nodes to broadcast messages to all other nodes. Consider the situation that $M'$ messages are to be broadcast by some nodes. On a network of arbitrary topology, we will need $O(\diam+M')$ rounds. On a congested clique, however, we need only $O(M'/n)$ rounds using the routing scheme of Dolev, Lenzen, and Peled~\cite[Lemma 1]{DolevLP12} (also see \cite{Lenzen13}):
If each node is source and destination of up to $n$ messages of size $O(\log n)$ (initially only the sources know destinations and contents of their messages), we will need $O(1)$ rounds to route the messages to their destinations. In particular, we can broadcast $n$ messages in $O(1)$ rounds, and thus $M'$ messages in $O(M'/n)$ rounds. Using this fact, the number of rounds needed for the algorithm in \Cref{sec:distributed hop set} reduces from $O(\sum_{v\in V'} |\clust (v, \cA, R, G')| + R\diam )$ on the overlay network~$ G' $ (cf.~\eqref{eq:cluster time}) to $O(\sum_{v\in V} |\clust (v, \cA, R, G)|/n + R) = \tilde O(pn^{1/p} + R) = n^{o(1)}$ on a congested clique~$ G $.\footnote{Instead of relying on the result of Dolev, Lenzen, and Peled, we can use the following algorithm to broadcast $M'$ messages in $O(M'/n)$ rounds. We assign an order to the messages, where messages sent by a node with smaller ID appear first in the order and messages sent by the same node appear in any order (a node can learn the order of its messages after it knows how many messages other nodes have). We then broadcast the first $n$ messages according to this order, say $M_1, \ldots, M_n$, where message $M_i$ is sent to a node with the $i$th smallest ID, and such a node sends $M_i$ to all other nodes. This takes only two rounds. The next messages are handled similarly. This algorithm broadcasts each $n$ messages using two rounds, and thus the total number of rounds is $O(M'/n)$.} 

Once we have an $(n^{o(1)}, o(1))$-hop set, we proceed as in \Cref{sec:distributed final}. Let $H = G \cup F$ be the graph obtained by adding to the input graph $G$ the edges of the $(n^{o(1)}, o(1))$-hop set $ F $. We can treat $H$ as a congested clique network with edge weights different from $G$. ($H$ can be computed without any additional communication since every node already knows the hop set.)
To $(1+o(1))$-approximate $\dist(s, v, G)$ for every node $v$ in $G$, it is sufficient to compute the $h$-hop distance $\dist^h(s, v, H)$, where $h=n^{o(1)}$. To do this, we follow the same approach for this problem as in \cite[Section 5.1]{Nanongkai-STOC14}, where we execute the distributed version of the Bellman--Ford algorithm for $h$ rounds. That is, every node $u$ maintains a tentative distance from the source $s$, denoted by $\delta(s, u)$, and in each round every node $u$ broadcasts $\delta(s, u)$ to all other nodes. It can be shown that after $k$ rounds every node $v$ knows the $k$-hop distance (i.e., $\delta(s, u) = \dist^k(s, v, H)$) correctly, and thus after $h$ rounds we will get the $h$-hop distances as desired.\footnote{Note that instead of the Bellman--Ford algorithm, one can also follow the steps in \Cref{sec:distributed final}. This gives a $(1+o(1))$-approximate value for $\dist^h(s, v, H)$ for every node $v$, which is sufficient for computing a $(1+o(1))$-approximate value for $\dist(s, v, G)$. This algorithm is, however, more complicated.}

\begin{theorem}
In the congested clique model, there is a deterministic algorithm that, on any weighted undirected clique network with polynomially bounded positive integer edge weights, computes $ (1 + o(1)) $-approximate shortest paths between a given source node $ s $ and every other node in $ n^{o(1)} $ rounds.
\end{theorem}

\subsection{Streaming Algorithm}

In the graph streaming model, the edges of the input graph are presented to the algorithm in an arbitrary order.
The goal is to design algorithms that process this ``stream'' of edges using as little space as possible.
In the multipass streaming model we are allowed to read the stream several times and want to keep both the number of passes and the amount of space used as small as possible.

Our streaming algorithm for constructing an $(n^{o(1)}, o(1))$-hop set proceeds in almost the same way as the distributed algorithm in \Cref{sec:distributed hop set}. First, observe that a shortest-path tree up to distance~$R$ can be computed in $O(R)$ passes and with $\tilde O(n)$ space: We use the space to remember the tentative distances of the nodes to $s$, and the shortest-path tree computed thus far. At the end of the $L$th pass we add nodes having distances exactly $L$ to the shortest-path tree and update the distance of their neighbors in the $(L+1)$th pass. 

We compute the priorities, as described in \Cref{sec:distributed priorities}, by solving $ p \leq \log{n} $ instances of an $(S, \gamma, \sigma)$-detection problem with $\gamma = R = N^{o(1)} $ and $ \sigma = q = N^{o(1)} $.
Observe that the guarantees of the source detection algorithm by Lenzen and Peleg for the broadcast \congest model directly carry over to the streaming model by simulating the algorithm as follows:
\begin{itemize}
\item The tentative list of each node is stored using $ O (\min{(\gamma, \wdiam)} + \min{(\sigma, |S|)}) $ space as, at any time, each node only needs to know at most $ \min{(\gamma, \wdiam)} + \min{(\sigma, |S|)} $ entries in its list (upper-bounded by the total number of messages each node will send).
\item The broadcast of one message per node in each round is simulated by writing the $ O (n) $ messages to space.
\item The reception of messages in each round is simulated by making a pass over the graph: Every time an edge $ (u, v) $ and its corresponding weight are read from the stream, the reception of $u$'s message by $ v $ is simulated by reading $u$'s message from space and then manipulating $v$'s tentative list accordingly.
\end{itemize}
We can summarize the guarantees of the source detection algorithm in the streaming model as follows.
\begin{theorem}[Implicit in \cite{LenzenP_podc13}]\label{lem:source detection algorithm streaming}
In the multipass streaming model, there is a deterministic algorithm for solving the $(S, \gamma, \sigma)$-detection problem in $ \min{(\gamma, \wdiam)} + \min{(\sigma, |S|)} $ passes with $O(n \cdot (\min{(\gamma, \wdiam)} + \min{(\sigma, |S|)}))$ space.
\end{theorem}
The algorithm for computing the priorities therefore needs $O(p (R + q))=n^{o(1)}$ passes and $O(n (R + q)) = n^{1+o(1)} $ space. 

To compute clusters, we compute $n$ shortest-path trees up to distance $R$ rooted at different nodes in parallel. The number of passes is clearly $O(R)$. The space is bounded by the sum of the sizes of the shortest-path trees. This is $O(\sum_{v\in V} |\clust (v, \cA, R, G)|)$, which, by \Cref{thm:clusters}, is $\tilde O(pn^{1+1/p}) = n^{1+o(1)}$. 
To compute the hop set we only have to compute clusters $ \tilde O (\log{W}) $ times. So, we need $ n^{o(1)} \log W $ passes and $ n^{1+o(1)} \log{W} $ space in total.
By considering the edges of the hop set in addition to the edges read from the stream, it suffices to compute approximate \sssp up to $ n^{o(1)} $ hops.
Using the streaming version of the Bellman--Ford algorithm (one pass per iteration), this can be done in $ n^{o(1)} \log{W} $ additional passes.

\begin{theorem}
In the multipass streaming model, there is a deterministic algorithm algorithm that, given any weighted undirected graph with polynomially bounded positive integer edge weights, computes $ (1 + o(1)) $-approximate shortest paths between a given source node $ s $ and every other node in $n^{o(1)} $ passes with $ n^{1+o(1)} $ space.
\end{theorem}

\section{Conclusion and Open Problems}\label{sec:conslusion}

We present deterministic distributed $(1+o(1))$-approximation algorithms for solving the \sssp problem on distributed weighted networks and other settings. The efficiencies of our algorithms match the known lower bounds up to an $n^{o(1)}$ factor. Important tools are a deterministic hop set construction and a deterministic process that replaces the well-known (randomized) hitting set argument.

In the conference version of this paper~\cite{HenzingerKN-STOC16}, we left as an open problem whether the factor of~$ n^{o(1)} $ in our bounds could be eliminated, and in particular we asked whether this can be done by constructing a $ (\polylog{n}, o(1)) $-hop set of size $ \tilde O (n) $.
Such a hop set construction, however can be ruled out by a recent lower bound of Abboud, Bodwin, and Pettie~\cite{AbboudBP17}.
Our open problem was solved nonetheless by Becker et al.~\cite{BeckerKKL16}, who, using tools from continuous optimization, showed that, in all the models that we considered above, a $ (1 + \epsilon)$-approximation can be obtained with an overhead of $ \epsilon^{-O(1)} \polylog{n} $ compared to known lower bounds.

Our deterministic replacement of the hitting set argument works only when the input graph is undirected. Our second open problem is thus how to derandomize algorithms on directed graphs (where edge directions do not affect the communication; see \cite{Nanongkai-STOC14,Nanongkai-SIROCCO16} for more details). In particular, it is known that \sssp can be $(1+\epsilon)$-approximated on directed weighted graphs in $\tilde O(\sqrt{n D} + D)$ time~\cite{Nanongkai-STOC14}, and single-source reachability can be computed in $\tilde O(\sqrt{n} D^{1/4}+D)$ time~\cite{GhaffariU15}. However, these results are obtained by {\em randomized} algorithms, and whether there are sublinear-time {\em deterministic} algorithms for these problems is still open. 

Finally, while our paper essentially settles the running time for computing single-source shortest paths approximately, the best running time for solving this problem {\em exactly} is $ O ((n \log{n})^{2/3} \diam^{1/3} + (n \log{n})^{5/6}) $~\cite{Elkin17}, a recent result obtained after the conference version of our paper appeared.
This leaves a gap to the $ \tilde \Omega (\sqrt{n} + \diam) $ lower bound~\cite{Elkin06}, and it is therefore natural to ask for an improved upper or lower bound.
In fact, in the past few years we have much better understood how to {\em approximately} solve basic graph problems, such as minimum cut, \sssp, all-pairs shortest paths, and maximum flows, on distributed networks (e.g., \cite{NanongkaiS14_disc,GhaffariK13,GhaffariKKLP15}). However, when it comes to solving these problems {\em exactly}, almost nothing is known. Understanding the complexity of exact algorithms is an important open problem. 

We refer the reader to \cite{Nanongkai-STOC14,Nanongkai-SIROCCO16} for further open problems.

\appendix
\section*{Appendix}
\section{Proof of \Cref{thm:property of weight rounding}}\label{sec:proof of property of weight rounding}

To prove~\eqref{eq:apsp approx main one}, let $ \pi_i $ be a shortest path between $ u $ and $ v $ in $ G_i $.
Observe that if we consider this path in $ G $ (with the corresponding edge weights), then its total weight is at least the distance between $ u $ and $ v $ in $ G $, i.e., $ w (\pi_i, G) \geq \dist (u, v, G) $, because no path in $ G $ can have weight less than the shortest path in $ G $.
We therefore get
\begin{multline*}
\rho_i \cdot \dist(u, v, G_i) = \rho_i \cdot \sum_{(x,y) \in \pi_i} w (x, y, G_i) = \sum_{(x,y) \in \pi_i} \rho_i \cdot \left\lceil \frac{w (x, y, G)}{\rho_i} \right\rceil \\
 \geq \sum_{(x,y) \in \pi_i} w (x, y, G) = w (\pi_i, G) \geq \dist (u, v, G) \, .
\end{multline*}

To prove~\eqref{eq:apsp approx main three}, let $ \pi $ be a shortest $h$-hop path from $ u $ to $ v $ in $ G $.
Observe that $ w (\pi, G_i) \geq \dist (u, v, G_i) $, as again no path has smaller weight than the shortest path in~$ G_i $.
By additionally exploiting the assumption $ \dist^h (u, v, G) \geq 2^i $, we get
\begin{align*}
\dist (u, v, G_i) \cdot \rho_i &\leq w (\pi, G_i) \cdot \rho_i = \sum_{(x, y) \in \pi} w (x, y, G_i) \cdot \rho_i
 = \sum_{(x, y) \in \pi} \left\lceil \frac{w (x, y, G)}{\rho_i} \right\rceil \cdot \rho_i \\
 &\leq \sum_{(x, y) \in \pi} (w (x, y, G) + \rho_i)
 = w (\pi, G) + | \pi | \cdot \rho_i
 = \dist^h (u, v, G) + | \pi | \cdot \rho_i \\
 &\leq \dist^h (u, v, G) + h \cdot \rho_i
 = \dist^h (u, v, G) + \epsilon 2^i
 \leq \dist^h (u, v, G) + \epsilon \dist^h (u, v, G) \\
 &= (1 + \epsilon) \dist^h (u, v, G) \, .
\end{align*}
To prove~\eqref{eq:apsp approx main two}, we combine~\eqref{eq:apsp approx main three} with the assumption $ \dist^h (u, v, G) \leq 2^{i+1} $:
\begin{equation*}
\dist (u, v, G_i) \leq \frac{(1 + \epsilon) \dist^h (u, v, G)}{\rho_i} = \frac{(1 + \epsilon) h \dist^h (u, v, G)}{\epsilon 2^i} \leq \frac{(1 + \epsilon) h 2^{i+1}}{\epsilon 2^i} = (2 + 2 / \epsilon) h \, .
\end{equation*}

\section{Proof of \Cref{lem:existence of hitting set}}\label{sec:proof of existence of hitting set}

We prove the claim by the probabilistic method.
Consider a sampling process that determines a set~$ T \subseteq U $ by adding each element of $ U $ to $ T $ independently with probability $ 1 / (2 x) $.
Let $ E_0 $ denote the event that $ | T | > |U| / x $, and for every $ 1 \leq j \leq k $ let $ E_j $ denote the event that $ T \cap S_j = \emptyset $.
First, observe that the size of $ T $ is $ |U| / (2 x) $ in expectation.
By Markov's inequality, we can bound the probability that the size of $ T $ is at most twice the expectation by at least $ 1/2 $ and thus $ \Pr [E_0] = \Pr [|T| > |U| / x] \leq 1/2 $.
Furthermore, for every $ 1 \leq j \leq k $, the probability that $ S_j $ contains no node of $ T $ is
\begin{equation*}
\Pr [E_j] = \left( 1 - \frac{1}{2x} \right)^{|S_j|} \leq 1 - \left( 1 - \frac{1}{2x} \right)^{2 x  \ln{3k}} \leq  \frac{1}{e^{\ln{3 k}}} = \frac{1}{3 k} \, .
\end{equation*}
The set $ T $ fails to have the desired properties of a small hitting set if at least one of the events $ E_j $ occurs.
By the union bound we have
\begin{equation*}
\Pr \left[ \bigcup_{0 \leq j \leq k} E_j \right] \leq \sum_{0 \leq j \leq k} \Pr [E_j] \leq \frac{1}{2} + k \cdot \frac{1}{3 k} = \frac{1}{2} + \frac{1}{3} < 1 \, .
\end{equation*}
It follows that the sampling process constructed a hitting set $ T $ for $ \cC = \{ S_1, \ldots, S_k \} $ of size at most $ | T | \leq |U| / x $ with nonzero probability.
Therefore a set $ T $ with these properties must really exist.
This finishes the proof of \Cref{lem:existence of hitting set}.

\section{Ruling Set Algorithm}\label{sec:ruling set algorithm}

For each node $v$, we represent its ID by a binary number $v_1 v_2 \ldots v_{\lambda}$. Initially, we set $T_0=U$. The algorithm proceeds for $b$ iterations. 

In the $i$th iteration, we consider $u_i$ for every node $u\in T_{i-1}$. If $u_i=0$, $v$ remains in $T_i$ and sends a ``beep'' message to every node within distance $c-1$. This takes $c-1$ rounds as beep messages from different nodes can be combined. If $u_i=1$, it remains in $T_i$ if there is no node $v\in T_{i-1}$ such that $\dist(u, v, G)\leq c$ and $v_i=0$; in other words, it remains in $T_i$ if it does not hear any beep after $c-1$ rounds. The output is $T=T_\lambda$. 
The running time of the above algorithm is clearly $O(c \lambda)=O(c\log n)$. Also, the distance between every pair of nodes in $T$ is at least $c$ since for every pair of nodes $u$ and~$v$ of distance less than $c$, there is an $i$ such that $u_i\neq v_i$, and in the $i$th iteration if both $u$ and~$v$ are in~$T_{i-1}$, then one of them will send a beep and the other one will not be in $T_i$.
Finally, it can be shown by induction that after the $i$th round every node in~$U$ is at distance at most $i$ from some node in $T_i$; thus it follows that every node in~$U$ is at distance at most $c \lambda$ from some node in $T$.

\section*{Acknowledgments}
The authors thank Michael Elkin, Stephan Friedrichs, and Christoph Lenzen for their comments and questions on the previous version of this paper. D.~Nanongkai thanks Michael Elkin for bringing the notion of hop set to his attention.
The authors thank the anonymous reviewers of SICOMP for their detailed comments.

\printbibliography[heading=bibintoc] % Make bibliography show up in table of contents

\end{document}